\documentclass[11pt]{article}
\usepackage[utf8]{inputenc}
\usepackage[english]{babel}
\usepackage[active]{srcltx}
\usepackage{hyperref}
\usepackage[all]{xy}
\usepackage{amsfonts}
\usepackage{amsthm}
\usepackage{enumerate}
\usepackage{mathrsfs}
\usepackage{multicol}
\usepackage[all]{xy}
\usepackage{amsfonts}
\usepackage{latexsym,amsmath,amssymb}
\usepackage[usenames,dvipsnames]{color}
\usepackage{pstricks}
\usepackage[absolute,overlay]{textpos}
\usepackage{graphics,wrapfig,times}
\usepackage{xfrac}
\usepackage[usenames]{color}

\usepackage{color}
\definecolor{DarkBlue}{rgb}{0.1,0.1,0.5}
\definecolor{Red}{rgb}{0.9,0.1,0.1}
\definecolor{Green}{rgb}{0.3,0.7,0.0}
\definecolor{green2}{rgb}{0.1,0.7,0.2}
\definecolor{blue2}{rgb}{0.0,0.6,0.7}
\definecolor{pink}{rgb}{1,0.0,1}
\definecolor{orange}{rgb}{0.9,0.0,0.1}

\newtheorem{theorem}{Theorem}

\newtheorem{proposition}{Proposition}

\newtheorem{definition}{Definition}

\renewcommand{\d}{\mathrm{d}}

\renewcommand{\d}{\mathrm{d}}
\newcommand{\W}{\mathcal{W}}
\newcommand{\derpar}[2]{\displaystyle\frac{\partial{#1}}{\partial{#2}}}

\newcommand{\Lag}{\mathcal{L}}
\newcommand{\Leg}{\mathcal{FL}}

\newcommand{\df}{\Omega}

\newcommand{\Tan}{\mathrm{T}}
\newcommand{\inn}{{\mathop{i}\nolimits}}
\newcommand{\Lie}{\mathop{\mathrm{L}}\nolimits}

\newcommand{\bal}{\begin{align*}}
\newcommand{\eal}{\end{align*}}
\def\beq{\begin{equation}}
\def\eeq{\end{equation}}
\def\bea{\begin{eqnarray}}
\def\eea{\end{eqnarray}}
\def\beann{\begin{eqnarray*}}
\def\eeann{\end{eqnarray*}}
\def\ben{\begin{enumerate}}
\def\een{\end{enumerate}}
\def\bit{\begin{itemize}}
\def\eit{\end{itemize}}

\newtheorem{definicio}{Definition}[section]
\newtheorem{resultat}{Result}[section]
\newtheorem{assum}{Assumption}[section]
\newtheorem{prop}{Proposition}[section]
\newtheorem{lema}{Lemma}[section]
\newtheorem{theo}{Theorem}[section]
\newcommand{\bd}{\begin{definicio} } 
\newcommand{\ed}{\end{definicio} } 
\newcommand{\bt}{\begin{theo} } 
\newcommand{\et}{\end{theo} } 
\newcommand{\bi}{\begin{itemize} } 
\newcommand{\ei}{\end{itemize} } 
\newcommand{\be}{\begin{enumerate} } 
\newcommand{\ee}{\end{enumerate} } 
\newcommand{\br}{\begin{resultat} } 
\newcommand{\er}{\end{resultat} } 
\newcommand{\ba}{\begin{assum} } 
\newcommand{\ea}{\end{assum} } 
\newcommand{\bl}{\begin{lema}}
\newcommand{\el}{\end{lema}}
\newcommand{\bp}{\begin{prop}}

\def\df{{\mit\Omega}}
\def\Lag{{\cal L}}

\def\d{{\rm d}}

\def\Tan{{\rm T}}
\def\Lie{\mathop{\rm L}\nolimits}
\def\inn{\mathop{i}\nolimits}
\def\Cinfty{{\rm C}^\infty}

\def\tabaddress#1{{\small\it\begin{tabular}[t]{c}#1
\\[1.2ex]\end{tabular}}}

\title{\sc Multisymplectic Formalism for Cubic Horndeski Theories}
\author{
{\sc  Mauricio Doniz\thanks{{\bf e}-{\it mail}:
   mauricio.doniz@upc.edu} }, \\
   \tabaddress{  Department of Mathematics, Universitat Politècnica de Catalunya   
   }\\
   {\sc Jordi Gaset\thanks{{\bf e}-{\it mail}:
   jordi.gaset@unir.net }},  \\
   \tabaddress{Escuela Superior de Ingeniería y Tecnología, Universidad Internacional de La Rioja
}}
   \date{\today \\
   }

\begin{document}

\maketitle

\begin{abstract}

We present the covariant multisymplectic formalism for the so-called cubic Horndeski theories and discuss the geometrical and physical interpretation of the constraints that arise in the unified Lagrangian-Hamiltonian approach. We analyse in more detail the covariant Hamiltonian formalism of these theories and we show that there are particular conditions that must be satisfied for the Poincaré-Cartan form of the Lagrangian to project onto $J^1\pi$. From this result, we study when a formulation using only  multimomenta is possible. We further discuss the implications of the general case, in which the projection onto $J^1\pi$ conditions are not met. 

\end{abstract}


\noindent\textbf{MSC\,2020 codes:}
{\sl Primary:} 53D42,70S05, 83C05;
 \\
\indent\indent\indent\indent\indent
{\sl Secondary:}
35R01, 53C15, 53C80, 53Z05.

\newpage

\tableofcontents

\newpage

\section{Introduction}
\label{Section1}

The multisymplectic formalism is a generalisation of symplectic geometry for field theories. It provides a covariant framework of the Lagrangian, Hamiltonian and the Hamilton-Jacobi formulations for field theories. The geometrical aspects of the multisymplectic formalism for first order theories, and its manifolds and forms have been studied in detail in \cite{Cantrijn1, Cantrijn2, Echeverria-Enriquez1,Carinena, Forger,gaset_variational_2016}. Throughout this work we go through the features and results of the formalism that we will need. To understand the formalism at a deeper level, \cite{Gimmsy,RR2009,Ryvkin} constitute a good starting point.

It is possible study field theories of up to second-order with this formalism. In particular, \cite{pere1} covers the most relevant features of the second-order multisymplectic formalism. However, third and higher order field theories do not have unique Poincaré-Cartan forms, although some efforts in that direction have been discussed in \cite{Cedric,vitagliano2010}. Besides this, more problems regarding the non-uniqueness of the geometrical structures appear in the definition of the Legendre maps associated with higher-order Lagrangians and also problems arise while trying to impose a multimomentum phase space for the Hamiltonian formalism of such theories.

The best way to overcome the aforementioned problems is to use the unified Lagrangian-Hamiltonian formalism. It was first introduced by Skinner and Rusk in \cite{Skinner1,Skinner2,Skinner3}, and the basic idea is to merge the Lagrangian and Hamiltonian formalisms into one. Even though Skinner and Rusk's idea tames some of the problems, still some arbitrary parameters that appear in the solutions of the higher-order field equations and in the definition of the Legendre maps must be fixed to guarantee its consistency. A modification of this framework that clarifies the choice of the jet and the multimomentum bundles and removes all ambiguity for second-order field theories was developed in \cite{pere2}, and this is the approach that we shall use throughout this work.
 
The unified Lagrangian-Hamiltonian formalism allows us to extract all the relevant physical information of a given system. First, we identify the geometry, manifolds and bundles of the theory and set the Lagrangian-Hamiltonian problem. For regular Lagrangians, the field equations that arise are compatible and have solutions on the jet-multimomentum bundle. This is not the case for singular Lagrangians, where it is needed to implement a constraint algorithm to be able to find the corresponding submanifolds of the jet-multimomentum bundle on which the field equations are compatible and have solutions. The constraint algorithm that we will apply was developed in \cite{deLeon}.

In regular first-order theories, the holonomy condition is recovered from the local coordinate expression of the field equations \cite{Echeverria}. That is not the case for second-order field theories (even for regular Lagrangians) and, therefore, it is required to imposed it \textit{a priori}.  Moreover, in the unified formalism the singularity of the Lagrangian appears also as constraints, in particular, as the definition of the Legendre transform. This is convenient as the implementation of the constraint algorithm is straightforward.

Another advantage of the unified Lagrangian-Hamiltonian formalism is that one can derive a covariant Hamiltonian formulation, as long as some regularity conditions are meet. Another common construction of a Hamiltonian formulation for field theories consists on performing a space $+$ time decomposition of the covariant Lagrangian formalism and then perform an instantaneous Legendre transform. This was originally performed by Arnowitt, Deser and Misner for General Relativity \cite{ADM}. This ADM-like approach, also called the instantaneous Hamiltonian formulation, has been studied from a geometric point of view \cite{gotay_multisymplectic_1991,Margalef}. We shall delve into the relation and equivalence between these Hamiltonian formalisms for theories of gravity in a future work. 

For all these reasons the multisymplectic formalism, and particularly the unified Lagrangian-Hamiltonian formalism, is suitable for studying singular second-order field theories such as certain string theories, the Korteweg-de Vries model and some of the most relevant theories of gravity, including General Relativity. 

It is known that General Relativity (GR) is one of the most successful theories in the history of physics. For over a hundred years now, it has been tested and shown to be the standard model of gravity. However, it is also known that it is is a low-energy effective theory, incomplete in the sense that it is  non-renormalisable \cite{Klauder}. There are different motivations for studying modified models of gravity. From the phenomenological perspective, the relatively recent detection of the first gravitational waves opens a new way of testing generalised models of gravity that predict something different from what GR predict \cite{Yunes}. The problem of the fine-tuned cosmological constant needed to explain the accelerated expansion of the universe in General Relativity is a strong incentive for physicists to explore generalised models of gravity as well. If we turn to a theoretical point of view, studying modified gravity grants a deeper understanding of GR.
Finally, mathematically speaking, understanding the geometric structure of generalised models of gravity could impose strong constraints on the theory that could be used by physicists to discard models or, on the other hand, make them turn their attention to a certain model. 

That being said, Horndeski's theory is an interesting candidate for being the generalisation of GR since it is the most general diffeomorphism invariant, scalar-tensor theory that leads to second order equations of motion \cite{Kobayashi}. It is strongly hyperbolic, at least at weak coupling, and therefore admits a well-posed initial value problem \cite{Kovacs1}. Most importantly, this type of theories are causal and hence allow for the existence of dynamical black holes \cite{Reall1} which could be potentially observed with the current techniques. 

We will be focusing on the construction of the multisymplectic formalism for the cubic subclass of Horndeski's theory. This subclass leads to strongly hyperbolic equations \cite{Kovacs2} and recently has gained attention among cosmologists due to the fact that this model can describe a non-singular bouncing universe \cite{Ijjas}. Besides the importance of this subclass of theories as a physical model, its covariant Hamiltonian formulation arises relevant geometric and physical consequences, as we shall discuss in the last section. 

The multisymplectic formalism of relevant theories of gravity can be found in the literature. Some of the most relevant examples are General Relativity \cite{GasetGR}, metric-affine gravity with \cite{Vey, capriotti_unified_2018} and without vielbein \cite{gaset_new_2019}, Lovelock Gravity \cite{Capriotti} and even Chern-Simons gravity and the bosonic string \cite{Arnoldo}.

The aim of the present work is to present the multisymplectic framework for the cubic subclass of Horndeski theories. A major feature of this formalism is that it provides a recipe for obtaining the Hamiltonian formulation of generalised theories of Gravity. Hence, we will first establish the geometric framework of the theory and introduce a suited change of co-ordinates that simplifies the calculations of the constraint algorithm. Finally we will show how to obtain the Hamiltonian formulation of the theory and briefly discuss its implications. 

All the manifolds are real, second countable and of class $C^{\infty}$. Manifolds and mappings are
assumed to be smooth. Sum over crossed repeated indices is understood. Comas denote partial derivatives and semicolons covariant derivatives.

\section{Setting up the problem}
\label{Section2}

In this section, we will introduce the geometrical structures and the manifolds and bundles that we need to construct the formalism for cubic Horndeski's theories. To do so, we will present the Horndeski's Lagrangian, explain its main features and justify the importance of the cubic subclass of Horndeski's theories. Finally we will set up the Lagrangian-Hamiltonian problem for this theories.


\subsection{Geometry, manifolds and bundles of the theory}
\label{Subsection2.2}

Let $M$ be an oriented 4-dimensional spacetime with coordinates $x^\mu$, $\mu=0,1,2,3$ and whose volume form is denoted by $\eta \in \df^4(M)$. A scalar field is a map $\phi: M\rightarrow \mathbb{R}$ or, equivalently, its graph is a section of the product bundle $M\times \mathbb{R}$ over $M$.

The covariant configuration bundle\footnote{Covariant configuration bundle will be just referred to as configuration bundle from now on. In contrast, there is an instantaneous configuration bundle that appears in the ADM-like formulation of the theory.} for this system is a fiber bundle
$\pi\colon E\rightarrow M$, with
$E$ being the manifold $\left(M\times\mathbb{R}\right)\times_MS_2^{3,1}(M)$ , where $S_2^{3,1}(M)$ denotes the bundle of symmetric covariant two-tensors $g$ of Lorentz signature $(-,+,+,+)$ acting on $T_xM$.

The adapted fiber coordinates in $E$ are $(x^\mu,g_{\alpha\beta},\phi)$,
($\mu,\alpha,\beta=0,1,2,3$), where $g_{\alpha\beta}$ and $\phi$ are the component functions of the metric and the scalar field respectively. The volume form satisfies $\eta=\d x^0\wedge\d x^1\wedge\d x^2\wedge\d x^3\equiv\d^4x$.
Provided that the metric is symmetric, it only has ten independent components and hence $\dim E=15$. We shall consider this when we sum over indices on the fiber, thus it is useful to establish a convention and order the indices as $0\leq\alpha\leq\beta\leq3$.\\

The 
\textsl{$k$th-order jet bundles} of the projection $\pi$, $J^k\pi$, ($k=1,2,3$);
which are the manifolds of the $k$-jets of local sections 
$\Psi \in \Gamma(\pi)$ are equivalence classes of local sections of $\pi$
\cite{book:Saunders89}.
Points in $J^k\pi$ are denoted by $j^k_x\psi$, 
with $x \in M$ and $\psi \in \Gamma(\pi)$ being a representative
of the equivalence class.
If $\Psi \in \Gamma(\pi)$, we denote the \textsl{$k$th prolongation} of $\Psi$ to $J^k\pi$ by
$j^k\Psi \in \Gamma(\bar{\pi}^k)$.
We have the following natural projections: if $r \leqslant k$,
$$
\begin{array}{rcl}
\pi^k_r \colon J^k\pi & \longrightarrow & J^r\pi \\
j^k_x\Psi & \longmapsto & j^r_x\Psi
\end{array}
\quad ; \quad
\begin{array}{rcl}
\pi^k \colon J^k\pi & \longrightarrow & E \\
j^k_x\Psi & \longmapsto & \Psi(x)
\end{array}
\quad ; \quad
\begin{array}{rcl}
\bar{\pi}^k \colon J^k\pi & \longrightarrow & M \\
j^k_x\Psi & \longmapsto & x
\end{array} \ .
$$
Observe that $\pi^s_r\circ\pi^k_s=\pi^k_r$, $\pi^k_0 = \pi^k$, 
$\pi^k_k=\textnormal{Id}_{J^k\pi}$, and $\bar{\pi}^k = \pi \circ \pi^k$.
The induced coordinates in $J^3\pi$ are 
$(x^\mu,\,g_{\alpha\beta},\phi,\,g_{\alpha\beta,\mu},\phi_{,\mu},\,g_{\alpha\beta,\mu\nu},\phi_{,\mu\nu},
\,g_{\alpha\beta,\mu\nu\lambda},\phi_{,\mu\nu\lambda})$, ($0\leq\mu\leq\nu\leq\lambda\leq3$). We shall use all the possible permutations, although only the ordered ones are actual co-ordinates.\\

Now we can explicitly write the total derivative $D_{\tau}$ in these local co-ordinates as

\begin{align*}
D_\tau=\derpar{}{x^\tau}
+\sum_{\substack{\alpha\leq\beta\\\mu\leq\nu\leq\lambda}} 
\Big(g_{\alpha\beta,\tau}\derpar{}{g_{\alpha\beta}}&+ g_{\alpha\beta,\mu\tau}\derpar{}{g_{\alpha\beta,\mu}}+
 g_{\alpha\beta,\mu\nu\tau}\derpar{}{g_{\alpha\beta,\mu\nu}}
+ g_{\alpha\beta,\mu\nu\lambda\tau}\derpar{}{g_{\alpha\beta,\mu\nu\lambda}}\nonumber\\
+\phi_{;\tau}\derpar{}{\phi}&+\phi_{,\mu\tau}\derpar{}{\phi_{,\mu}}+\phi_{,\mu\nu\tau}\derpar{}{\phi_{,\mu\nu}}+\phi_{,\mu\nu\lambda\tau}\derpar{}{\phi_{,\mu\nu\lambda}}\Big).
\end{align*}

Notice that, if $f\in\Cinfty(J^k\pi)$, then $D_\tau f\in\Cinfty(J^{k+1}\pi)$.

Next, consider the bundle $J^1\pi$ and let ${\cal M}\pi\equiv\Lambda_2^4(J^{1}\pi)$ 
be the bundle of $4$-forms over $J^{1}\pi$ vanishing 
by the action of two $\bar{\pi}^{1}$-vertical vector fields; 
with the canonical projections
$$
\pi_{J^{1}\pi} \colon \Lambda_2^4(J^{1}\pi) \to J^{1}\pi \quad ; \quad
\bar{\pi}_M= \bar{\pi}^{1} \circ \pi_{J^{1}\pi} \colon \Lambda_2^4(J^{1}\pi) \to M \, .
$$

The induced local coordinates in $\Lambda_2^4(J^{1}\pi)$ are
$(x^\mu,g_{\alpha\beta},\phi,g_{\alpha\beta,\mu},\phi_{,\mu},p_g,p_{\phi},p_g^{\alpha\beta,\mu},p_{\phi}^{,\mu},\overline{p}_g^{\alpha\beta,\mu\nu},\overline{p}_{\phi}^{,\mu\nu})$, with $0\leq\alpha\leq\beta\leq3$ and $\mu,\nu=0,1,2,3$.

This bundle is endowed with the
\textsl{tautological (or Liouville) $4$-form} 
$\Theta_{1}\in\df^4(\Lambda_2^4(J^{1}\pi))$
and the \textsl{canonical (or Liouville) $5$-form} 
$\Omega_{1} = -d\Theta_{1} \in \Omega^5(\Lambda_2^4(J^{1}\pi))$, which is a multisymplectic form; meaning it is closed and $1$-nondegenerate,
with the following local expressions
\beann
\Theta_1&=& 
p\d^4x+\sum_{\alpha\leq\beta}\left(p_g^{\alpha\beta,\mu}\d g_{\alpha\beta}\wedge \d^{3}x_\mu+
\overline{p}_g^{\alpha\beta,\mu\nu}\d g_{\alpha\beta,\mu}\wedge \d^{3}x_{\nu}\right)+\overline{p}_{\phi}d\phi+\overline{p}_{\phi}^{,\mu}d\phi_{,\mu}
 \ , \\
\Omega_1&=&
-\d p\wedge \d^4x-\sum_{\alpha\leq\beta}\left(\d p_g^{\alpha\beta,\mu}\wedge \d g_{\alpha\beta}\wedge \d^{3}x_\mu+
\d \overline{p}_g^{\alpha\beta,\mu\nu}\wedge \d g_{\alpha\beta,\mu}\wedge \d^{3}x_{\nu}\right) - \d \overline{p}_{\phi}\wedge d\phi - \d \overline{p}_{\phi}^{,\mu}\wedge d\phi_{,\mu}\ ;
\eeann
where $\displaystyle \d^3x_\nu=\inn\left(\derpar{}{x^\nu}\right)\d^4x$.

Consider the $\pi_{J^{1}\pi}$-transverse submanifold
$\jmath_s \colon J^2\pi^\dagger \hookrightarrow \Lambda_2^4(J^1\pi)$ 
defined locally by the constraints 
$\overline{p}_g^{\alpha\beta,\mu\nu}=\overline{p}_g^{\alpha\beta,\nu\mu}$ and $\overline{p}_{\phi}^{,\mu\nu}=\overline{p}_{\phi}^{,\nu\mu}$.
This submanifold is called the \textsl{extended $2$-symmetric multimomentum bundle} and its construction is canonical even though it is defined in local coordinates \cite{art:Saunders_Crampin90}. 

Let 
$$
\pi_{J^1\pi}^\dagger \colon J^2\pi^\dagger \to J^1\pi
\quad ; \quad
\bar{\pi}_M^\dagger = \bar{\pi}^{1} \circ \pi_{J^1\pi}^\dagger \colon J^2\pi^\dagger \to M
$$
be the canonical projections. 

The coordinates in $J^2\pi^\dagger$ are 

\begin{equation*}
   (x^\mu,g_{\alpha\beta},\phi,g_{\alpha\beta,\mu},\phi_{,\mu},p,p_g^{\alpha\beta,\mu},p_{\phi}^{,\mu},p_g^{\alpha\beta,\mu\nu},p_{\phi}^{,\mu\nu}), 
\end{equation*}

with ($0\leq\alpha\leq\beta\leq 3$; $0\leq\mu\leq\nu\leq3$), and 
$\displaystyle\jmath_s^*\overline{p}_g^{\alpha\beta,\mu\nu}=\frac{1}{n(\mu\nu)}p_g^{\alpha\beta,\mu\nu}$, $\displaystyle\jmath_s^*\overline{p}_{\phi}^{,\mu\nu}=\frac{1}{n(\mu\nu)}p_{\phi}^{,\mu\nu}$, where $n(\mu\nu)$ is a combinatorial factor defined by
$n(\mu\nu)=1$ for $\mu=\nu$, and $n(\mu\nu)=2$ for $\mu\neq\nu$.

Denote $\Theta_1^s =\jmath_s^*\Theta_1 \in \Omega^4(J^2\pi^\dagger)$
and the multisymplectic form 
$\Omega_1^s= \jmath_s^*\Omega_1 = -d\Theta_1^s \in \Omega^5(J^2\pi^\dagger)$, 
which are called \textsl{symmetrised Liouville $m$ and $(m+1)$-forms}. In this case, the local expressions are 
\begin{align*}
\Theta_1^s=& p\,\d^4x+\sum_{\alpha\leq\beta}p_g^{\alpha\beta,\mu}\d g_{\alpha\beta}\wedge \d^{3}x_\mu+\sum_{\alpha\leq\beta}\frac{1}{n(\mu\nu)}p_g^{\alpha\beta,\mu\nu}\d g_{\alpha\beta,\mu}\wedge \d^{3}x_{\nu} \nonumber\\
&+ p_{\phi}^{,\mu}\d \phi \wedge \d^3 x_{\mu}+\frac{1}{n(\mu\nu)}p_{\phi}^{,\mu\nu}\d \phi_{\mu} \wedge \d^3 x_{\nu}\ , \\
\Omega_1^s=&-\d p\wedge \d^4x-\sum_{\alpha\leq\beta}\d p_g^{\alpha\beta,\mu}\wedge \d g_{\alpha\beta}\wedge \d^{3}x_\mu-\sum_{\alpha\leq\beta}\frac{1}{n(\mu\nu)}\d p_g^{\alpha\beta,\mu\nu}\wedge \d g_{\alpha\beta,\mu}\wedge \d^{3}x_{\nu}\nonumber\\
&-\d p_{\phi}^{,\mu}\wedge\d \phi \wedge \d^3 x_{\mu}-\frac{1}{n(\mu\nu)}\d p_{\phi}^{,\mu\nu}\wedge \d \phi_{\mu} \wedge \d^3 x_{\nu}\ .
\end{align*}

Last, consider the quotient bundle $J^2\pi^\ddagger = J^2\pi^\dagger / \Lambda^4_1(J^1\pi)$,
called the \textsl{restricted $2$-symmetric multimomentum bundle}, 
endowed with the following projections
$$
\mu \colon J^2\pi^\dagger \to J^2\pi^\ddagger\quad ; \quad
\pi_{J^1\pi}^\ddagger \colon J^2\pi^\ddagger \to J^1\pi\quad ; \quad
\bar{\pi}_M^\ddagger \colon J^2\pi^\ddagger \to M .
$$
Observe that $J^2\pi^\ddagger$ is also the submanifold of
$\Lambda^4_2(J^1\pi) / \Lambda^4_1(J^1\pi)$ 
defined by the local constraints $\overline{p}^{\alpha\beta,\mu\nu} 
= \overline{p}^{\alpha\beta,\nu\mu} $ and  $\overline{p}_{\phi}^{;\mu\nu} 
= \overline{p}_{\phi}^{;\nu\mu} $
Thus, the coordinates in $J^2\pi^\ddagger$ are 

\begin{equation*}
(x^\mu,g_{\alpha\beta},\phi,g_{\alpha\beta,\mu},\phi_{,\mu},p_g^{\alpha\beta,\mu},p_{\phi}^{,\mu},p_g^{\alpha\beta,\mu\nu},p_{\phi}^{,\mu\nu}),
\end{equation*}

with ($0\leq\alpha\leq\beta\leq 3$; $0\leq\mu\leq\nu\leq3$).
The dimension of this submanifold is $\dim J^{2}\pi^\ddagger = \dim J^{2}\pi^\dagger - 1$.

The {\sl Horndeski Lagrangian density} is a 
$\overline{\pi}^2$-semibasic form $\mathcal{L}_{\mathfrak V}\in\Omega^4(J^2\pi)$. Hence
$\mathcal{L}_{\mathfrak V}=L_{\mathfrak V}\,(\overline{\pi}^2)^*\eta$, 
where $L_{\mathfrak V}\in\Cinfty(J^2\pi)$ is the {\sl Horndeski Lagrangian function} 

\begin{equation*}
L_{\mathfrak V}=\frac{1}{16\pi G}\sqrt{|g|}\left(\sum_{i=1}^3L_i\right), 
\end{equation*}

where, in the coordinates of the manifold:
\begin{align*}
   L_1=&R+X\,;\quad L_2=G_2(\phi,X)\,;\quad L_3=G_3(\phi,X)\square \phi\,;
   \\
   L_4=&G_4(\phi,X)R+G_{4,X}(\phi,X)\left[(\square \phi)^2-g^{\mu\alpha}g^{\nu\beta}\phi_{;\mu\nu}\phi_{;\alpha\beta}\right]
\\
L_5=&G_5(\phi,X)G_{\mu\nu}g^{\mu\alpha}g^{\nu\beta}\phi_{;\alpha\beta}\nonumber\\
&-\frac{1}{6}G_{5,X}(\phi,X)\left[(\square \phi)^3+2g^{\nu\alpha}g^{\beta\gamma}g^{\sigma\nu}\phi_{;\mu\nu}\phi_{;\alpha\beta}\phi_{;\gamma\sigma}-3g^{\alpha\mu}g^{\beta\nu}g^{}\phi_{;\mu\nu}\phi_{;\alpha\beta}\square \phi\right].
\end{align*}

Here $g^{\mu\nu}$ are the components of the inverse of the metric tensor, that is, $g_{\alpha\mu}g^{\mu\nu}=\delta_\alpha^\nu$;
$g$ the determinant of the metric tensor; $X=-\frac{1}{2}g^{\mu\nu}\phi_{;\mu}\phi_{;\nu}$, and $\square \phi=g^{\mu\nu}\phi_{;\mu\nu}$. The Ricci tensor is given by: $R_{\alpha\beta}=D_\gamma\Gamma^{\gamma}_{\alpha\beta}-D_\alpha\Gamma^{\gamma}_{\gamma\beta}+
\Gamma^{\gamma}_{\alpha\beta}\Gamma^{\delta}_{\delta\gamma}-
\Gamma^{\gamma}_{\delta\beta}\Gamma^{\delta}_{\alpha\gamma}$. Hence, the Ricci scalar has the local expression: $$R=g^{\alpha\beta}R_{\alpha\beta}=g^{\alpha\beta}\left(D_\gamma\Gamma^{\gamma}_{\alpha\beta}-D_\alpha\Gamma^{\gamma}_{\gamma\beta}+
\Gamma^{\gamma}_{\alpha\beta}\Gamma^{\delta}_{\delta\gamma}-
\Gamma^{\gamma}_{\delta\beta}\Gamma^{\delta}_{\alpha\gamma}\right),$$

As we mentioned in the introduction, throughout this work we shall consider the cubic subclass of Horndeski's theories, which means we will only preserve up to the $L_3$ term in the Lagrangian, i.e. we will set $G_4(\phi,X)=0$ and $G_5(\phi,X)=0$.



\subsection{The higher-order jet multimomentum bundles}
\label{Subsection2.3}

The unified Lagrangian-Hamiltornian formalism is set in a bundle that encompasses the jets and bundles described in the previous section and hence the manifolds $M$ and $E$. First, we construct the \textsl{symmetric higher-order jet multimomentum bundle} $\mathcal{W}$ and the \textsl{restricted symmetric higher-order jet multimomentum bundle} $\mathcal{W}_r$ as described in \cite{pere1, pere2}

\begin{equation*}
    \mathcal{W}=J^3\pi\times_{J^1\pi}J^2\pi^\dagger,
\end{equation*}

\begin{equation*}
    \W_r = J^{3}\pi \times_{J^{1}\pi} \, J^{2}\pi^\ddagger.
\end{equation*}

Here $J^2\pi^\dagger$ $J^2\pi^\ddagger$ are the extended and the restricted 2-symmetric multimomentum bundle respectively, as discussed in the previous section. The symmetric higher-order jet multimomentum bundles have the following natural local coordinates 
\begin{equation*}
    (x^\mu,\,g_{\alpha\beta},\phi,\,g_{\alpha\beta,\mu},\phi_{,\mu},\,g_{\alpha\beta,\mu\nu},\phi_{,\mu\nu},
\,g_{\alpha\beta,\mu\nu\lambda},\phi_{,\mu\nu\lambda},p,p_g^{\alpha\beta,\mu},p_{\phi}^{,\mu},p_g^{\alpha\beta,\mu\nu},p_{\phi}^{,\mu\nu}),
\end{equation*}
and
\begin{equation*}
    (x^\mu,\,g_{\alpha\beta},\phi,\,g_{\alpha\beta,\mu},\phi_{,\mu},\,g_{\alpha\beta,\mu\nu},\phi_{,\mu\nu},
\,g_{\alpha\beta,\mu\nu\lambda},\phi_{,\mu\nu\lambda},p_g^{\alpha\beta,\mu},p_{\phi}^{,\mu},p_g^{\alpha\beta,\mu\nu},p_{\phi}^{,\mu\nu}),
\end{equation*}
with ($0\leq\alpha\leq\beta\leq 3$; $0\leq\mu\leq\nu\leq3$), and are endowed with the following projections
\beann
\rho_1 \colon \W \to J^{3}\pi \ ,\
&\rho_2 \colon \W \to J^{2}\pi^\dagger \ ,\ &
\rho_M \colon \W \to M
\\
\rho^r_1 \colon \W_r \to J^{3}\pi \ ,\
&\rho^r_2 \colon \W_r \to J^{2}\pi^\ddagger  \ ,\ &
\rho_M^r \colon \W_r \to M \ .
\eeann

Moreover, the quotient map
$\mu \colon J^{2}\pi^\dagger \to J^{2}\pi^\ddagger$
induces a natural submersion $\mu_\W \colon \W \to \W_r$.

Now, we define the canonical pairing which will help us determine the Hamiltonian function. 
$$
\begin{array}{rcl}
\mathcal{C} \colon J^{2}\pi \times_{J^{1}\pi} \Lambda_2^4(J^{1}\pi) & \longrightarrow & \Lambda_1^4(J^{1}\pi) \\
(j^{2}_x\phi,\omega) & \longmapsto & (j^{1}\phi)^*_{j^{1}_x\phi}\omega
\end{array} \ ,
$$
hence we have can define a new pairing $\mathcal{C}^s \colon J^2\pi \times_{J^1\pi} J^2\pi^\dagger \to \Lambda_1^4(J^1\pi)$ as
$$
\mathcal{C}^s(j^{2}_x\phi,\omega) = \mathcal{C}(j^{2}_x\phi,j_s(\omega)) = (j^1\phi)_{j^1_x\phi}^* \ j_s(\omega) \, .
$$

From here we get the \textsl{second-order coupling $4$-form} in $\mathcal{W}$, which is the $\rho_M$-semibasic $4$-form
$\hat{\mathcal{C}} \in \Omega^4(\mathcal{W})$  defined by
$$
\hat{\mathcal{C}}(j^3_x\phi,\omega) = \mathcal{C}^s(\pi^3_2(j^3_x\phi),\omega) 
\quad ,\quad  (j^3_x\phi,\omega) \in \mathcal{W} \ .
$$
As $\hat{\mathcal{C}}$ is a $\rho_M$-semibasic $4$-form, 
there exists a function $\hat{C} \in C^\infty(\mathcal{W})$ such that 
$\hat{\mathcal{C}} = \hat{C}\rho_M^*\eta$. In co-ordinates this is written as
$$
\hat{\mathcal{C}}= 
\left(p+\sum_{\alpha\leq\beta}p_g^{\alpha\beta,\mu}g_{\alpha\beta,\mu}
+\sum_{\substack{\alpha\leq\beta\\\mu\leq\nu}}p_g^{\alpha\beta,\mu\nu}g_{\alpha\beta,\mu\nu}+p_{\phi}^{,\mu}\phi_{,\mu}+\sum_{\mu\leq\nu}p_{\phi}^{,\mu\nu}\phi_{,\mu\nu}\right)\d^4x \ .
$$

A $4-$form $\hat{\mathcal{L}}=(\pi^3_2 \circ \rho_1)^*\mathcal{L}_{\mathfrak V}\in \Omega^4(\mathcal{W})$,
which can be written as $\hat{\mathcal{L}} = \hat{L} \,\rho_M^*\eta$,
where $\hat{L} = (\pi^3_2 \circ \rho_1)^*L_{\mathfrak V}\in C^\infty(\mathcal{W})$, can be used to define the {\sl Hamiltonian submanifold} $\mathcal{W}_o$
$$
\mathcal{W}_o = \left\{ w \in \mathcal{W} \colon \hat{\mathcal{L}}(w) = \hat{\mathcal{C}}(w) \right\} \stackrel{\jmath_o}{\hookrightarrow} \mathcal{W} \ ,
$$
which is ultimately defined by the following constraint
$$
\hat{\mathcal{C}}-\hat{L}\equiv p+\sum_{\alpha\leq\beta}p_g^{\alpha\beta,\mu}g_{\alpha\beta,\mu}
+\sum_{\substack{\alpha\leq\beta\\\mu\leq\nu}}p_g^{\alpha\beta,\mu\nu}g_{\alpha\beta,\mu\nu}+p_{\phi}^{,\mu}\phi_{,\mu}+\sum_{\mu\leq\nu}p_{\phi}^{,\mu\nu}\phi_{,\mu\nu}-\hat{L} = 0 \ .
$$

This submanifold is $\mu_\mathcal{W}$-transverse and diffeomorphic to $\mathcal{W}_r$, $\Phi_o\colon\W_o\to\W_r$.
$\mathcal{W}_o$ induces a \textsl{Hamiltonian section}
$\hat{h} \in \Gamma(\mu_\mathcal{W})$ by
$\hat{h} = \jmath_o \circ \Phi_o^{-1} \colon \mathcal{W}_r \to \mathcal{W}$, 
specified by the local \textsl{Hamiltonian function}
$$
\hat H=\sum_{\alpha\leq\beta}p_g^{\alpha\beta,\mu}g_{\alpha\beta,\mu}+\sum_{\substack{\alpha\leq\beta\\\mu\leq\nu}}p_g^{\alpha\beta,\mu\nu}g_{\alpha\beta,\mu\nu}+p_{\phi}^{,\mu}\phi_{,\mu}+\sum_{\mu\leq\nu}p_{\phi}^{,\mu\nu}\phi_{,\mu\nu}-\hat L \ .
$$
that is, 
\beann
\hat{h}
(x^\mu,\,g_{\alpha\beta},\phi,\,g_{\alpha\beta,\mu},\phi_{,\mu},\,g_{\alpha\beta,\mu\nu},\phi_{,\mu\nu},
\,g_{\alpha\beta,\mu\nu\lambda},\phi_{,\mu\nu\lambda},p_g^{\alpha\beta,\mu},p_{\phi}^{,\mu},p_g^{\alpha\beta,\mu\nu},p_{\phi}^{,\mu\nu}) 
\\
=(x^\mu,\,g_{\alpha\beta},\phi,\,g_{\alpha\beta,\mu},\phi_{,\mu},\,g_{\alpha\beta,\mu\nu},\phi_{,\mu\nu},
\,g_{\alpha\beta,\mu\nu\lambda},\phi_{,\mu\nu\lambda},-\hat{H},p_g^{\alpha\beta,\mu},p_{\phi}^{,\mu},p_g^{\alpha\beta,\mu\nu},p_{\phi}^{,\mu\nu}).
\eeann

This is all summarised in the following commutative diagram

$$
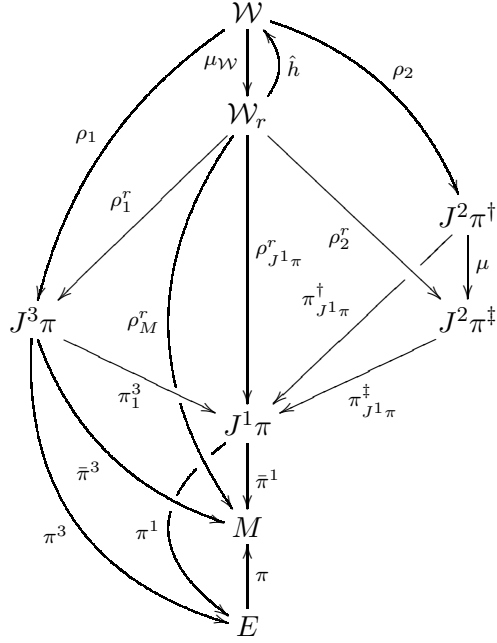
\xymatrix{
\ & \ & \mathcal{W} \ar@/_1.3pc/[llddd]_{\rho_1} \ar[d]_-{\mu_\mathcal{W}} \ar@/^1.3pc/[rrdd]^{\rho_2} & \ & \ \\
\ & \ & \mathcal{W}_r \ar@/_1pc/[u]_{\hat{h}} \ar[lldd]_{\rho_1^r} \ar[rrdd]_{\rho_2^r} \ar[ddd]^<(0.4){\rho_{J^{1}\pi}^r} \ar@/_2.5pc/[dddd]_-{\rho_M^r}|(.675){\hole} & \ & \ \\
\ & \ & \ & \ & J^{2}\pi^\dagger \ar[d]^-{\mu} \ar[lldd]_{\pi_{J^{1}\pi}^\dagger}|(.25){\hole} \\
J^{3}\pi \ar[rrd]_{\pi^{3}_{1}} \ar@/_1.3pc/[ddrr]_-{\bar{\pi}^{3}} \ar@/_2.5pc/[rrddd]_{\pi^{3}}& \ & \ & \ & J^{2}\pi^\ddagger \ar[dll]^{\pi_{J^{1}\pi}^\ddagger} \\
\ & \ & J^{1}\pi \ar[d]^{\bar{\pi}^{1}} \ar@/_2.5pc/[dd]_-{\pi^1}|(.165){\hole}|(.375){\hole}& \ & \ \\
\ & \ & M & \ & \ \\
\ & \ & E \ar[u]_-{\pi}& \ & \
}
$$

The {\sl Liouville forms } in $\W_r$,
$\Theta_r = (\rho_2 \circ \hat{h})^*\Theta_1^s \in \df^4(\W_r)$ and
$\Omega_r=-\d\Theta_r=(\rho_2 \circ \hat{h})^*\Omega_1^s\in\df^5(\W_r)$,
for second order field theories \cite{pere2}, in these specific co-ordinates, are

\begin{align*}
\Theta_r=&
-\hat{H}\d^4x+\sum_{\alpha\leq\beta}p^{\alpha\beta,\mu}\d g_{\alpha\beta}\wedge \d^{3}x_\mu
+\sum_{\alpha\leq\beta}\frac{1}{n(\mu\nu)}p^{\alpha\beta,\mu\nu}\d g_{\alpha\beta,\mu}\wedge \d^{3}x_{\nu}\nonumber\\
&+p^{,\mu\nu}\d \phi \wedge \d^{3}x_{\mu}+\frac{1}{n(\mu\nu)} p^{,\mu\nu} \d \phi_{,\mu} \wedge \d^3 x_{\nu}
\end{align*}

\begin{align}\label{eqn:premultisymplecticform}
\Omega_r=&
\d \hat{H}\wedge \d^4x-\sum_{\alpha\leq\beta}\d p_g^{\alpha\beta,\mu}\wedge \d g_{\alpha\beta}\wedge \d^{3}x_\mu-\sum_{\alpha\leq\beta}\frac{1}{n(\mu\nu)}\d p_g^{\alpha\beta,\mu\nu}\wedge \d g_{\alpha\beta,\mu}\wedge \d^{3}x_{\nu}\nonumber\\
&-\d p_g^{,\mu\nu}\wedge\d \phi \wedge \d^{3}x_{\mu}-\frac{1}{n(\mu\nu)} \d p_g^{,\mu\nu}\wedge \d \phi_{,\mu} \wedge \d^3 x_{\nu};
\end{align}

To obtain the form (\ref{eqn:premultisymplecticform}) explicitly, we need to calculate the exterior derivative of the Hamiltonian function. The exterior derivative of a function is the differential of the function. Specifically we have 

\begin{align*}
    \d\hat H=&\sum_{\alpha\leq\beta}\left(g_{\alpha\beta,\mu}\d p_g^{\alpha\beta,\mu}+p_g^{\alpha\beta,\mu}\d g_{\alpha\beta,\mu}\right)+\sum_{\substack{\alpha\leq\beta\\\mu\leq\nu}}\left(g_{\alpha\beta,\mu\nu}\d p_g^{\alpha\beta,\mu\nu}+p_g^{\alpha\beta,\mu\nu}\d g_{\alpha\beta,\mu\nu}\right)\\
    &+\phi_{,\mu}\d p_{\phi}^{,\mu}+p_{\phi}^{,\mu}\d \phi_{,\mu}+\sum_{\mu\leq\nu}\left(\phi_{,\mu\nu}\d p_{\phi}^{,\mu\nu}+p_{\phi}^{,\mu\nu}\d \phi_{,\mu\nu}\right)-\d \hat L.
\end{align*}

The differential of the Lagrangian, provided its dependency upon the metric, the first and second order derivatives of the metric, the scalar field and the first and second order derivatives of the scalar field, is

\begin{align*}
    \d \hat{L}=&\sum_{\alpha\leq\beta}\derpar{\hat L}{g_{\alpha\beta}}\d g_{\alpha\beta}+\sum_{\alpha\leq\beta}\derpar{\hat L}{g_{\alpha\beta,\mu}}\d g_{\alpha\beta,\mu}+\sum_{\substack{\alpha\leq\beta\\\mu\leq\nu}}\derpar{\hat L}{g_{\alpha\beta,\mu\nu}}\d g_{\alpha\beta,\mu\nu}\\
    &+\derpar{\hat L}{\phi}\d\phi+\derpar{\hat L}{\phi_{,\mu}}\d\phi_{,\mu}+\derpar{\hat L}{\phi_{,\mu\nu}}\d\phi_{,\mu\nu}.
\end{align*}

The Liouville forms are degenerate; this is
\beq
\ker\,\Theta_r=\ker\,\Omega_r\supset
{\left<\frac{\partial}{\partial g_{\alpha\beta,\mu\nu}},\frac{\partial}{\partial g_{\alpha\beta,\mu\nu\lambda}},\frac{\partial}{\partial \phi_{,\mu\nu\lambda}}\right>}_{0\leq\alpha\leq\beta\leq3;\, 0\leq\mu\leq\nu\leq\lambda\leq3}\ .
\label{gaugevf}
\eeq
For a premultisymplectic form $\Omega$,
we call {\sl (geometric) gauge vector fields} to those vector fields belonging to $\ker\,\Omega$ (see \cite{gaset_variational_2016,gaset_geometric_2022} for more details).
Furthermore, $\Theta_r$ is $(\pi^3_1\circ\rho^r_2)$-projectable.






\subsection{The Lagrangian-Hamiltonian problem}
\label{laghamunif}

Consider the system $(\W_r,\Omega_r)$.

\begin{definition}
A section $\psi \in \Gamma(\bar{\pi}^{k})$ is {\rm holonomic} if
$j^k(\pi^{k} \circ \psi) = \psi$; that is, $\psi$ is the $k$th prolongation of a section
$\phi = \pi^{k} \circ \psi \in \Gamma(\pi)$,
and an integrable and $\bar{\pi}_M$-transverse multivector field 
${\bf X} \in \mathfrak{X}^4(J^k\pi)$ is
{\rm holonomic} if its integral sections are holonomic.

A section $\psi \in \Gamma(\bar{\pi}_M^\ddagger)$ is {\rm holonomic in 
$J^2\pi^\ddagger$} if
$\bar{\pi}_{J^1\pi}^\ddagger\circ\psi\in\Gamma(\bar{\pi}^{1})$ is holonomic in $J^1\pi$,
and an integrable and $\bar{\pi}_M^\ddagger$-transverse multivector field 
${\bf X} \in \mathfrak{X}^4(J^2\pi^\ddagger)$ is
{\rm holonomic} if its integral sections are holonomic.

Finally, a section $\psi \in \Gamma(\rho_M^r)$ is {\rm holonomic in $\W_r$} if
$\rho_1^r \circ \psi \in \Gamma(\bar{\pi}^{3})$ is holonomic in $J^{3}\pi$,
and an integrable and $\rho_M^r$-transverse multivector field ${\bf X} \in \mathfrak{X}^4(\mathcal{W}_r)$ is
{\rm holonomic} if its integral sections are holonomic.
\end{definition}

It is important to point out that the fact that a multivector field in $\W_r$ 
has the local expression \eqref{locholmv}
(and then being locally decomposable and $\rho^r_M$-transverse)
is just a necessary condition to be holonomic, since it may not be integrable.
However, if such a multivector field admits integral sections, 
then its integral sections are holonomic.
In general, a locally decomposable and $\rho^r_M$-transverse multivector field
which has \eqref{locholmv} as coordinate expression, is said to be {\sl semiholonomic} in $\W_r$.

The \textsl{Lagrangian-Hamiltonian problem} associated with the system $(\W_r,\Omega_r)$
consists in finding holonomic sections $\psi \in \Gamma(\rho_M^r)$ satisfying
any of the following equivalent conditions:
\begin{enumerate}
\item $\psi$ is a solution to the equation
\begin{equation}\label{eqn:UnifFieldEqSect}
\psi^*\inn(X)\Omega_r = 0 \, , \quad \mbox{for every } X \in \mathfrak{X}(\mathcal{W}_r) \ .
\end{equation}
\item $\psi$ is an integral section of a multivector field contained in a class of holonomic multivector fields $\{ {\bf X} \} \subset \mathfrak{X}^4(\mathcal{W}_r)$
satisfying the equation
\begin{equation}\label{eqn:UnifDynEqMultiVF}
\inn({\bf X})\Omega_r = 0 \ .
\end{equation}
\end{enumerate}

As the form $\Omega_r$ is $1$-degenerate we have that
$(\mathcal{W}_r,\Omega_r)$ is a premultisymplectic system,
and solutions to (\ref{eqn:UnifFieldEqSect}) or (\ref{eqn:UnifDynEqMultiVF})
do not exist everywhere in $\mathcal{W}_r$. 

\section{A suitable change of coordinates}
\label{Section3}

At this point, for the sake of facilitating calculations, it is useful to introduce a new set of co-ordinates on $\W_r$. The argument is simple; most interesting Lagrangians contain covariant derivative terms, making it necessary to transform them to terms with partial derivatives and components of the Levi-Civita connection. However, if we choose the velocity, acceleration and jerk co-ordinates of the scalar field on $\W_r$ to be the covariant derivatives of first, second and third order respectively, instead of the partial derivatives, we will reduce significantly the difficulty in further computations. The cost of this co-ordinate transformation is that the Poincaré-Cartan forms must be transformed to the new co-ordinates.

The suited coordinates on $\W_r$ in this case are:

\begin{equation*}
    (\tilde{x}^\mu,\,\tilde{g}_{\alpha\beta},\tilde{\phi},\,\tilde{g}_{\alpha\beta,\mu},\tilde{\phi}_{,\mu},\,\tilde{g}_{\alpha\beta,\mu\nu},\tilde{\phi}_{,\mu\nu},
\,\tilde{g}_{\alpha\beta,\mu\nu\lambda},\tilde{\phi}_{,\mu\nu\lambda},\tilde{p}_g^{\alpha\beta,\mu},\tilde{p}_{\phi}^{,\mu},\tilde{p}_g^{\alpha\beta,\mu\nu},\tilde{p}_{\phi}^{,\mu\nu}).
\end{equation*}

They are related to the previous chart by a set of maps. For the coordinates of spacetime we just set $\tilde{x}^\mu=x^\mu$. The metric is left unchanged related, and the scalar field is changed as:
\begin{align*}
    \tilde{\phi}=&\phi\,;\qquad
    \tilde{\phi}_{;\mu}=\phi_{,\mu}\,;\qquad
    \tilde{\phi}_{;\mu\nu}=\phi_{,\mu\nu}-\phi_{,\gamma}\Gamma_{\mu\nu}^{\gamma}\,;
    \\
    \tilde{\phi}_{;\mu\nu\lambda}=&\phi_{,\mu\nu\lambda}-\phi_{,\gamma\lambda}\Gamma^{\gamma}_{\nu\mu}-\phi_{,\gamma}\Gamma^{\gamma}_{\nu\mu,\lambda}-\phi_{,\gamma\nu}\Gamma^{\gamma}_{\mu\lambda}+\phi_{,\sigma}\Gamma^{\sigma}_{\gamma\nu}\Gamma^{\gamma}_{\mu\lambda}-\phi_{,\mu\gamma}\Gamma^{\gamma}_{\lambda\nu}+\phi_{,\sigma}\Gamma^{\sigma}_{\mu\gamma}\Gamma^{\gamma}_{\lambda\nu}.
\end{align*}

The multimomenta coordinates are mapped via the identity, and they are also unchanged. These expressions are a mere change of coordinates and hence the manifold and all the geometric structure of the theory remains intact. From these relations we get 

\begin{align*}
    \frac{\partial}{\partial g_{\alpha\beta}}=&\frac{\partial}{\partial \tilde{g}_{\alpha\beta}}+\frac{\partial \tilde{\phi}_{;\mu'\nu'}}{\partial g_{\alpha\beta}}\frac{\partial}{\partial \tilde{\phi}_{;\mu'\nu'}}+\frac{\partial \tilde{\phi}_{;\mu'\nu'\lambda'}}{\partial g_{\alpha\beta}}\frac{\partial}{\partial \tilde{\phi}_{;\mu'\nu'\lambda'}},\\
    \frac{\partial}{\partial g_{\alpha\beta,\mu}}=&\frac{\partial}{\partial \tilde{g}_{\alpha\beta,\mu}}+\frac{\partial \tilde{\phi}_{;\mu'\nu'}}{\partial g_{\alpha\beta,\mu}}\frac{\partial}{\partial \tilde{\phi}_{;\mu'\nu'}}+\frac{\partial \tilde{\phi}_{;\mu'\nu'\lambda'}}{\partial g_{\alpha\beta,\mu}}\frac{\partial}{\partial \tilde{\phi}_{;\mu'\nu'\lambda'}},\\
    \frac{\partial}{\partial g_{\alpha\beta,\mu\nu}}=&\frac{\partial}{\partial \tilde{g}_{\alpha\beta,\mu\nu}}+\frac{\partial \tilde{\phi}_{;\mu'\nu'\lambda'}}{\partial g_{\alpha\beta,\mu\nu}}\frac{\partial}{\partial \tilde{\phi}_{;\mu'\nu'\lambda'}},\\
    \frac{\partial}{\partial \phi} =& \frac{\partial}{\partial \tilde{\phi}}\\
    \frac{\partial}{\partial \phi_{,\mu}}=&\frac{\partial}{\partial \tilde{\phi}_{;\mu}}+\frac{\partial \tilde{\phi}_{;\mu'\nu'}}{\partial \phi_{,\mu}}\frac{\partial}{\partial \tilde{\phi}_{;\mu'\nu'}}+\frac{\partial \tilde{\phi}_{;\mu'\nu'\lambda'}}{\partial \phi_{,\mu}}\frac{\partial}{\partial \tilde{\phi}_{;\mu'\nu'\lambda'}},\\
    \frac{\partial}{\partial \phi_{,\mu\nu}}=&\frac{\partial}{\partial \tilde{\phi}_{;\mu\nu}}+\frac{\partial \tilde{\phi}_{;\mu'\nu'\lambda'}}{\partial \phi_{,\mu\nu}}\frac{\partial}{\partial \tilde{\phi}_{;\mu'\nu'\lambda'}},
\end{align*}

where 

\begin{align*}
    \frac{\partial \tilde{\phi}_{;\mu'\nu'}}{\partial g_{\alpha\beta}}=&n(\alpha\beta)\tilde{\phi}_{;\gamma}g^{\gamma(\alpha}\Gamma_{\mu'\nu'}^{\beta)}\\
    \frac{\partial \tilde{\phi}_{;\mu'\nu'}}{\partial g_{\alpha\beta,\mu}}=&\frac{n(\alpha\beta)}{2}\tilde{\phi}_{;\gamma}g^{\gamma\rho}\left[\delta_{\mu'}^{\mu}\delta_{\nu'}^{(\alpha}\delta_{\rho}^{\beta)}+\delta_{\nu'}^{\mu}\delta_{\rho}^{(\alpha}\delta_{\mu'}^{\beta)}-\delta_{\rho}^{\mu}\delta_{\mu'}^{(\alpha}\delta_{\nu'}^{\beta)}\right]\\
    \frac{\partial \tilde{\phi}_{;\mu'\nu'}}{\partial \phi_{,\mu}}=&-\Gamma_{\mu'\nu'}^{\mu}\\
    \frac{\partial \tilde{\phi}_{;\mu'\nu'\lambda'}}{\partial \phi_{,\mu}}=&-\Gamma_{\mu'\nu',\lambda'}^{\mu}+\Gamma_{\gamma\nu'}^{\mu}\Gamma_{\mu'\lambda'}^{\gamma}+\Gamma_{\mu'\gamma}^{\mu}\Gamma_{\lambda'\nu'}^{\gamma}\\
    \frac{\partial \tilde{\phi}_{;\mu'\nu'\lambda'}}{\partial \phi_{,\mu\nu}}=&-3n(\mu\nu)\Gamma_{[\nu'\mu'}^{(\mu}\delta^{\nu)}_{\lambda']}\\
    \frac{\partial \tilde{\phi}_{;\mu'\nu'\lambda'}}{\partial g_{\alpha\beta}}=&n(\alpha\beta)\Bigg\{-\left(\phi_{;\gamma\lambda'}+\phi_{;\tau}\Gamma_{\gamma\lambda'}^{\tau}\right)g^{\gamma(\alpha}\Gamma^{\beta)}_{\nu'\mu'}\nonumber\\
    &+\frac{1}{2}\phi_{;\gamma}\Big[g^{\gamma(\alpha}g^{\beta)\rho}(g_{\mu'\rho,\nu'\lambda'}+g_{\rho\nu',\mu'\lambda'}-g_{\mu'\nu',\rho\lambda'})\nonumber\\
    &-g^{\gamma(\alpha}g^{\beta)\sigma}g^{\rho\delta}g_{\sigma\delta,\lambda'}(g_{\mu'\rho,\nu'}+g_{\rho\nu',\mu'}-g_{\nu'\mu',\rho})\nonumber\\
    &-g^{\gamma\sigma}g^{\rho(\alpha}g^{\beta)\delta}g_{\sigma\delta,\lambda'}(g_{\mu'\rho,\nu'}+g_{\rho\nu',\mu'}-g_{\nu'\mu',\rho})\Big]\nonumber\\
    &+(\phi_{;\gamma\nu'}+\phi_{;\tau}\Gamma_{\gamma\nu'}^{\tau})g^{\gamma(}\Gamma^{\beta)}_{\mu'\lambda'}-\phi_{;\sigma}g^{\sigma(\alpha}\Gamma^{\beta)}_{\gamma\nu'}\Gamma^{\gamma}_{\mu'\sigma'}\nonumber\\
    &-\phi_{;\sigma}\Gamma^{\sigma}_{\gamma\nu'}g^{\gamma(\alpha}\Gamma^{\beta)}_{\mu'\sigma'}+(\phi_{;\mu'\gamma}+\phi_{;\tau}\Gamma^{\tau}_{\mu'\gamma})g^{\gamma(\alpha}\Gamma^{\beta)}_{\gamma'\nu'}\nonumber\\
    &-\phi_{;\sigma}g^{\sigma(\alpha}\Gamma^{\beta)}_{\mu'\gamma}\Gamma^{\gamma}_{\lambda'\nu'}-\phi_{;\sigma}\Gamma^{\sigma}_{\mu'\gamma}g^{\gamma(\alpha}\Gamma^{\beta)}_{\nu'\lambda'}\Bigg\}\\
    \frac{\partial \tilde{\phi}_{;\mu'\nu'\lambda'}}{\partial g_{\alpha\beta,\mu}}=&-\frac{1}{2}n(\alpha\beta)\Big[g^{\gamma\sigma}g^{\rho\delta}\delta^{\mu}_{\lambda'}\delta_{\lambda'}^{(\alpha}\delta^{\beta)}_{\delta}(g_{\mu'\rho,\nu'}+g_{\rho\nu',\mu'}-g_{\nu'\mu',\rho})\nonumber\\
    &+g^{\gamma\sigma}g^{\rho\delta}g_{\sigma\delta,\lambda'}(\delta_{\nu'}^{\mu}\delta_{\mu'}^{(\alpha}\delta_{\rho}^{\beta)}+\delta_{\mu'}^{\mu}\delta_{\rho}^{(\alpha}\delta_{\nu'}^{\beta)}-\delta_{\rho}^{\mu}\delta_{\nu'}^{(\alpha}\delta_{\mu'}^{\beta)})\Big]\\
    \frac{\partial \tilde{\phi}_{;\mu'\nu'\lambda'}}{\partial g_{\alpha\beta,\mu\nu}}=&\frac{1}{2}n(\alpha\beta)n(\mu\nu)g^{\gamma\rho}\left(\delta_{\mu'}^{(\alpha}\delta_{\rho}^{\beta)}\delta_{\nu'}^{(\mu}\delta_{\sigma'}^{\nu)}+\delta_{\rho}^{(\alpha}\delta_{\nu'}^{\beta)}\delta_{\mu'}^{(\mu}\delta_{\sigma'}^{\nu)}-\delta_{\nu'}^{(\alpha}\delta_{\mu'}^{\beta)}\delta_{\rho}^{(\mu}\delta_{\sigma'}^{\nu)}\right)
\end{align*}

It is important to clarify that we have used the standard notation for symmetrization and antisymmetrization of indices. To symmetrize on $n$ indices, we sum over all possible permutations of these indices and divide the result by $n!$ . To antisymmetrize, we go through the same procedure, but weighting each term in the sum by the sign of the permutation. For instance

\begin{align*}
    T^{(\mu \nu \rho)\lambda}=&\frac{1}{3!}\left(T^{\mu \nu \rho \lambda}+T^{\nu \rho \mu \lambda}+T^{\rho\mu \nu \lambda}+T^{\nu\mu \rho \lambda}+T^{\rho \nu\mu \lambda}+T^{\mu \rho\nu \lambda}\right),\\
    T^{\mu}_{\;\;\;[\nu cd]}=&\frac{1}{3!}\left(T^{\mu}_{\;\;\;\nu \rho \sigma}-T^{\mu}_{\;\;\;\rho\nu \sigma}+T^{\mu}_{\;\;\;\rho \sigma\nu}-T^{\mu}_{\;\;\;\sigma\rho\nu}+T^{\mu}_{\;\;\;\sigma\nu \rho}-T^{\mu}_{\;\;\;\nu \sigma\rho}\right).
\end{align*}

Sometimes it is convenient to (anti)symmetrize over indices which are not adjacent. In this case, we use vertical bars to denote that some indices will be excluded. For instance,

\begin{equation*}
    T_{(\mu|\rho\sigma|\nu)}=\frac{1}{2}\left(T_{\mu\rho\sigma\nu}+T_{\nu\rho\sigma\mu}\right).
\end{equation*}

From now on, we shall drop the tildes over the names of the co-ordinates, but it must be understood that the following calculations are performed in the new chart. 

The Liouville forms in $\W_r$, in these new coordinates, become

\begin{align*}
\Theta_r=&
-\hat{H}\d^4x+\sum_{\alpha\leq\beta}p^{\alpha\beta,\mu}\d g_{\alpha\beta}\wedge \d^{3}x_\mu
+\sum_{\alpha\leq\beta}\frac{1}{n(\mu\nu)}p^{\alpha\beta,\mu\nu}\d g_{\alpha\beta,\mu}\wedge \d^{3}x_{\nu}\nonumber\\
&+p^{,\mu\nu}\d \phi \wedge \d^{3}x_{\mu}+\frac{1}{n(\mu\nu)} p^{,\mu\nu}_{\phi} \d \phi_{;\mu} \wedge \d^3 x_{\nu}
\end{align*}

\begin{align*}
\Omega_r=&
\d \hat{H}\wedge \d^4x-\sum_{\alpha\leq\beta}\d p_g^{\alpha\beta,\mu}\wedge \d g_{\alpha\beta}\wedge \d^{3}x_\mu-\sum_{\alpha\leq\beta}\frac{1}{n(\mu\nu)}\d p_g^{\alpha\beta,\mu\nu}\wedge \d g_{\alpha\beta,\mu}\wedge \d^{3}x_{\nu}\nonumber\\
&-\d p_{\phi}^{,\mu\nu}\wedge\d \phi \wedge \d^{3}x_{\mu}-\frac{1}{n(\mu\nu)} \d p_{\phi}^{,\mu\nu}\wedge \d \phi_{;\mu} \wedge \d^3 x_{\nu};
\end{align*}

where

\begin{align*}
    \hat H=&\sum_{\alpha\leq\beta}p_g^{\alpha\beta,\mu} g_{\alpha\beta,\mu}+\sum_{\substack{\alpha\leq\beta \\   \mu\leq\nu}}p_g^{\alpha\beta,\mu\nu} g_{\alpha\beta,\mu\nu}+p_{\phi}^{,\mu} \phi_{;\mu}+\sum_{\mu\leq\nu}p_{\phi}^{,\mu\nu} \phi_{;\mu\nu}+\sum_{\mu\leq\nu}p_{\phi}^{,\mu\nu} \phi_{;\gamma} \Gamma^{\gamma}_{\mu\nu}-\d \hat L
\end{align*}

\begin{align*}
    \d\hat H=&\sum_{\alpha\leq\beta}p_g^{\alpha\beta,\mu}\d g_{\alpha\beta,\mu}+\sum_{\alpha\leq\beta}g_{\alpha\beta,\mu}\d p_g^{\alpha\beta,\mu}+\sum_{\alpha\leq\beta}p_g^{\alpha\beta,\mu\nu}\d g_{\alpha\beta,\mu\nu}+\sum_{\alpha\leq\beta}g_{\alpha\beta,\mu\nu}\d p_g^{\alpha\beta,\mu\nu}\nonumber\\
    &+p_{\phi}^{,\mu}\d \phi_{;\mu}+\phi_{;\mu}\d p_{\phi}^{,\mu}+\sum_{\mu\leq\nu}p_{\phi}^{,\mu\nu}\d \phi_{;\mu\nu}+\sum_{\mu\leq\nu}p_{\phi}^{,\mu\nu}\Gamma^{\gamma}_{\mu\nu}\d \phi_{;\gamma}+\sum_{\mu\leq\nu}p_{\phi}^{,\mu\nu} \phi_{;\gamma} \d\Gamma^{\gamma}_{\mu\nu}\nonumber\\
    &+\sum_{\mu\leq\nu}\left(\phi_{;\mu\nu}+\phi_{;\gamma}\Gamma^{\gamma}_{\mu\nu}\right)\d p_{\phi}^{,\mu\nu}-\d \hat L,
\end{align*}

with

\begin{equation*}
    \d \Gamma_{\mu\nu}^{\gamma}=-\frac{1}{2}g^{\gamma\lambda}g^{\rho\sigma}\left(g_{\nu\rho,\mu}+g_{\rho\mu,\nu}-g_{\mu\nu,\rho}\right)\d g_{\lambda\sigma}+\frac{1}{2}g^{\gamma\rho}\left(\d g_{\nu\rho,\mu}+\d g_{\rho\mu,\nu}-\d g_{\mu\nu,\rho}\right).
\end{equation*}

Provided that we are interested in second order theories, we have

\begin{align*}
    \d \hat{L}=& \frac{\partial \hat{L}}{\partial x^{\mu}}\d x^{\mu}+\frac{\partial \hat{L}}{\partial g_{\alpha\beta}} \d g_{\alpha\beta}+\frac{\partial \hat{L}}{\partial \tilde{g}_{\alpha\beta,\mu}} \d g_{\alpha\beta,\mu}+\frac{\partial \hat{L}}{\partial \tilde{g}_{\alpha\beta,\mu\nu}} \d g_{\alpha\beta,\mu\nu}\nonumber\\
    &+\frac{\partial \hat{L}}{\partial \phi} \d \phi+\frac{\partial \hat{L}}{\partial \phi_{;\mu}} \d \phi_{;\mu}+\frac{\partial \hat{L}}{\partial \phi_{;\mu\nu}} \d \phi_{;\mu\nu}.
\end{align*}

The local expression  of a holonomic multivector field 
${\bf X} \in \mathfrak{X}^4(\mathcal{W}_r)$ in the new coordinates is

\bea
\mathbf{X} &=& 
\bigwedge_{\lambda=0}^3\sum_{\substack{\alpha\leq\beta\\\mu\leq\nu\leq\tau}}
\left[\derpar{}{x^\lambda}+g_{\alpha\beta,\lambda}\left(\frac{\partial}{\partial \tilde{g}_{\alpha\beta}}+\frac{\partial \tilde{\phi}_{;\mu'\nu'}}{\partial g_{\alpha\beta}}\frac{\partial}{\partial \tilde{\phi}_{;\mu'\nu'}}+\frac{\partial \tilde{\phi}_{;\mu'\nu'\lambda'}}{\partial g_{\alpha\beta}}\frac{\partial}{\partial \tilde{\phi}_{;\mu'\nu'\lambda'}}\right)
\right.\nonumber \\
& &
\qquad\qquad\qquad\left. 
+
g_{\alpha\beta,\mu\lambda}\left(\frac{\partial}{\partial \tilde{g}_{\alpha\beta,\mu}}+\frac{\partial \tilde{\phi}_{;\mu'\nu'}}{\partial g_{\alpha\beta,\mu}}\frac{\partial}{\partial \tilde{\phi}_{;\mu'\nu'}}+\frac{\partial \tilde{\phi}_{;\mu'\nu'\lambda'}}{\partial g_{\alpha\beta,\mu}}\frac{\partial}{\partial \tilde{\phi}_{;\mu'\nu'\lambda'}}\right)
\right.\nonumber \\
& &
\qquad\qquad\qquad\left. 
+g_{\alpha\beta,\mu\nu\lambda}\left(\frac{\partial}{\partial \tilde{g}_{\alpha\beta,\mu\nu}}+\frac{\partial \tilde{\phi}_{;\mu'\nu'\lambda'}}{\partial g_{\alpha\beta,\mu\nu}}\frac{\partial}{\partial \tilde{\phi}_{;\mu'\nu'\lambda'}}\right)+\phi_{,\lambda}\derpar{}{\phi}
\right.\nonumber \\
& &
\qquad\qquad\qquad\left.
+\left(\phi_{;\mu\lambda}+\phi_{;\gamma}\Gamma_{\mu\lambda}^{\gamma}\right)\left(\frac{\partial}{\partial \tilde{\phi}_{;\mu}}+\frac{\partial \tilde{\phi}_{\mu'\nu'}}{\partial \phi_{,\mu}}\frac{\partial}{\partial \tilde{\phi}_{;\mu'\nu'}}+\frac{\partial \tilde{\phi}_{;\mu'\nu'\lambda'}}{\partial \phi_{,\mu}}\frac{\partial}{\partial \tilde{\phi}_{;\mu'\nu'\lambda'}}\right)
\right.\nonumber \\
& &
\left.
+\left(\phi_{;\mu\nu\lambda}+\phi_{;\gamma\lambda}\Gamma_{\nu\mu}^{\gamma}-\phi_{;\rho}\Gamma_{\gamma\lambda}^{\rho}\Gamma_{\mu\nu}^{\gamma}+\phi_{;\gamma}\Gamma_{\nu\mu,\lambda}^{\gamma}+\phi_{;\gamma\nu}\Gamma_{\mu\lambda}^{\gamma}+\phi_{;\mu\gamma}\Gamma_{\lambda\nu}^{\gamma}\right)\left(\frac{\partial}{\partial \tilde{\phi}_{;\mu\nu}}+\frac{\partial \tilde{\phi}_{;\mu'\nu'\lambda'}}{\partial \phi_{,\mu\nu}}\frac{\partial}{\partial \tilde{\phi}_{;\mu'\nu'\lambda'}}\right)
\right.\nonumber \\
& &
\qquad\qquad\qquad\left. 
+F_{g\;\alpha\beta,\mu\nu\tau\lambda}\derpar{}{g_{\alpha\beta,\mu\nu\tau}}+
G^{\alpha\beta,\mu}_{g\;\lambda}\derpar{}{p_g^{\alpha\beta,\mu}}+ 
G^{\alpha\beta,\mu\nu}_{g\;\lambda}\derpar{}{p_g^{\alpha\beta,\mu\nu}}
\right.\nonumber \\
& &
\qquad\qquad\qquad\left. 
+F_{\phi\;,\mu\nu\tau\lambda}\derpar{}{\phi_{,\mu\nu\tau}}+
G^{,\mu}_{\phi\;\lambda}\derpar{}{p_{\phi}^{,\mu}}+ 
G^{,\mu\nu}_{\phi\;\lambda}\derpar{}{p_{\phi}^{,\mu\nu}} \right] \ .
\label{locholmv}
\eea

and, if 
\begin{align*}
  \psi(x^\lambda) =   (x^\lambda,\,&\psi_{g\;\alpha\beta}(x^\lambda),
\,\psi_{g\;\alpha\beta,\mu}(x^\lambda),\,\psi_{g\;\alpha\beta,\mu\nu}(x^\lambda),
\,\psi_{\phi}(x^\lambda),\,\psi_{\phi\;;\mu}(x^\lambda),\,\psi_{\phi\;;\mu\nu}(x^\lambda),
\\\nonumber
&\psi_{g\;\alpha\beta,\mu\nu\tau}(x^\lambda),\,\psi_g^{\alpha\beta,\mu}(x^\lambda),\psi_g^{\alpha\beta,\mu\nu}(x^\lambda)\,\psi_{\phi\;,\mu\nu\tau}(x^\lambda),\,\psi_{\phi}^{,\mu}(x^\lambda),\psi_\phi^{,\mu\nu}(x^\lambda))
\end{align*}
is an integral section of ${\bf X}$, its component functions
satisfy the following system of partial differential equations
\beann
\derpar{\psi_{g\;\alpha\beta}}{x^\lambda}=g_{\alpha\beta,\lambda}\circ\psi \ , \
\derpar{\psi_{g\;\alpha\beta,\mu}}{x^\lambda}=g_{\alpha\beta,\mu\lambda}\circ\psi \ , \
\derpar{\psi_{g\;\alpha\beta,\mu\nu}}{x^\lambda}=g_{\alpha\beta,\mu\nu\lambda}\circ\psi \ , 
\nonumber\\
\derpar{\psi_{\phi}}{x^\lambda}=\phi_{;\lambda}\circ\psi \ , \
\derpar{\psi_{\phi\;;\mu}}{x^\lambda}=\phi_{;\mu\lambda}\circ\psi \ , \
\derpar{\psi_{\phi\;;\mu\nu}}{x^\lambda}=\phi_{;\mu\nu\lambda}\circ\psi \ , \
\nonumber\\
\derpar{\psi_{g\;\alpha\beta,\mu\nu\tau}}{x^\lambda}=F_{g\;\alpha\beta,\mu\nu\tau\lambda}\circ\psi \ , \
\derpar{\psi_g^{\alpha\beta,\mu}}{x^\lambda}=G^{\alpha\beta,\mu}_{g\;\lambda}\circ\psi  \ , \
\derpar{\psi_g^{g\;\alpha\beta,\mu\nu}}{x^\lambda}=G^{\alpha\beta,\mu\nu}_{g\;\lambda}\circ\psi  \nonumber\\
\derpar{\psi_{\phi\;,\mu\nu\tau}}{x^\lambda}=F_{\phi\;\alpha\beta,\mu\nu\tau\lambda}\circ\psi \ , \
\derpar{\psi_{\phi}^{,\mu}}{x^\lambda}=G^{,\mu}_{\phi\;\lambda}\circ\psi  \ , \
\derpar{\psi_{\phi}^{,\mu\nu}}{x^\lambda}=G^{,\mu\nu}_{\phi\;\lambda}\circ\psi  \ .
\label{pdesect}
\eeann

\section{Unified Lagrangian-Hamiltonian formalism and the Constraint Algorithm}
\label{Section4}

Now we explicitly calculate the Legendre maps and the corresponding field equations for the multivector fields in the new coordinates.

\subsection{Legendre maps}

\begin{proposition}
\label{prop:GraphLegMapSect}
A section $\psi \in \Gamma(\rho_M^r)$ solution to the equation \eqref{eqn:UnifFieldEqSect} 
takes values in a $140$-codimensional submanifold 
$\jmath_{\Lag_{\mathfrak V}}\colon\mathcal{W}_{\mathcal{L}_{\mathfrak V}} \hookrightarrow \mathcal{W}_r$ which is identified with the graph of
a bundle map $\mathcal{FL}_{\mathfrak V}\colon J^3\pi \to J^{2}\pi^\ddagger$, over $J^1\pi$, defined locally by

\begin{align*}
    \mathcal{FL}_{\mathfrak V}^{\ \ *}\,p_g^{\alpha\beta,\mu}=&
\frac{\partial \hat L}{\partial g_{\alpha\beta,\mu}} - 
\sum_{\nu=0}^{3}\frac{1}{n(\mu\nu)}
X_\nu\left( \frac{\partial \hat L}{\partial g_{\alpha\beta,\mu\nu}}\right)-\frac{1}{2}\left(\frac{\partial \hat{L}}{\partial \phi_{;\mu\alpha}}\phi_{;\gamma}g^{\gamma\beta}+\frac{\partial \hat{L}}{\partial \phi_{;\beta\mu}}\phi_{;\gamma}g^{\gamma\alpha}-\frac{\partial \hat{L}}{\partial \phi_{;\alpha\beta}}\phi_{;\gamma}g^{\gamma\mu}\right)\nonumber\\
=&\hat L_g^{\alpha\beta,\mu},\\
\mathcal{FL}_{\mathfrak V}^{\ \ *}\,p_g^{\alpha\beta,\mu\nu}=&
\frac{\partial \hat L}{\partial g_{\alpha\beta,\mu\nu}},\\
\mathcal{FL}_{\mathfrak V}^{\ \ *}\,p_{\phi}^{,\mu}=&
\frac{\partial \hat L}{\partial \phi_{;\mu}} - 
\sum_{\nu=0}^{3}\frac{1}{n(\mu\nu)}
X_\nu\left( \frac{\partial \hat L}{\partial \phi_{;\mu\nu}}\right)-\frac{\partial \hat{L}}{\partial \phi_{;\gamma\nu}}\Gamma_{\gamma\nu}^{\mu}=
\hat L_{\phi}^{,\mu},\\
\mathcal{FL}_{\mathfrak V}^{\ \ *}\,p_{\phi}^{,\mu\nu}=&
\frac{\partial \hat L}{\partial \phi_{;\mu\nu}}.
\label{Legmap}
\end{align*}

The submanifold 
$\mathcal{W}_{\mathcal{L}_{\mathfrak V}}$ is the graph of a bundle morphism
$\widetilde{\mathcal{FL}}_{\mathfrak V} \colon J^3\pi \to J^2\pi^\dagger$ over $J^1\pi$ defined locally by

\begin{align*}
    \displaystyle\widetilde{\mathcal{FL}}_{\mathfrak V}^{\ \ *}\,p_g^{\alpha\beta,\mu}=&
\frac{\partial \hat L}{\partial g_{\alpha\beta,\mu}} - 
\sum_{\nu=0}^{3}\frac{1}{n(\mu\nu)}
X_\nu\left( \frac{\partial \hat L}{\partial g_{\alpha\beta,\mu\nu}}\right)\nonumber\\
&-\frac{1}{2}\left(\frac{\partial \hat{L}}{\partial \phi_{;\mu\alpha}}\phi_{;\gamma}g^{\gamma\beta}+\frac{\partial \hat{L}}{\partial \phi_{;\beta\mu}}\phi_{;\gamma}g^{\gamma\alpha}-\frac{\partial \hat{L}}{\partial \phi_{;\alpha\beta}}\phi_{;\gamma}g^{\gamma\mu}\right)\nonumber\\
=&\hat L_g^{\alpha\beta,\mu},\\
\displaystyle\widetilde{\mathcal{FL}}_{\mathfrak V}^{\ \ *}\,p_g^{\alpha\beta,\mu\nu}=&
\frac{\partial \hat L}{\partial g_{\alpha\beta,\mu\nu}},\nonumber\\
\displaystyle\widetilde{\mathcal{FL}}_{\mathfrak V}^{\ \ *}\,p_{\phi}^{,\mu}=&
\frac{\partial \hat L}{\partial \phi_{;\mu}} - 
\sum_{\nu=0}^{3}\frac{1}{n(\mu\nu)}
X_\nu\left( \frac{\partial \hat L}{\partial \phi_{;\mu\nu}}\right)-\frac{\partial \hat{L}}{\partial \phi_{;\gamma\nu}}\Gamma_{\gamma\nu}^{\mu}=
\hat L_{\phi}^{,\mu},\\
\displaystyle\widetilde{\mathcal{FL}}_{\mathfrak V}^{\ \ *}\,p_{\phi}^{,\mu\nu}=&
\frac{\partial \hat L}{\partial \phi_{;\mu\nu}},\\
\displaystyle\widetilde{\mathcal{FL}}_{\mathfrak V}^{\ \ *}\,p=&\hat{L}-\sum_{\substack{\alpha\leq\beta\\\mu\leq\nu}}g_{\alpha\beta,\mu\nu}\frac{\partial \hat L}{\partial g_{\alpha\beta,\mu\nu}}-\sum_{\substack{\mu\leq\nu}}\left(\phi_{;\mu\nu}+\phi_{;\gamma}\Gamma^{\gamma}_{\mu\nu}\right)\frac{\partial \hat L}{\partial \phi_{;\mu\nu}}\nonumber\\
&-\phi_{;\mu}\left(\frac{\partial \hat L}{\partial \phi_{;\mu}} - 
\sum_{\nu=0}^{3}\frac{1}{n(\mu\nu)}
X_\nu\left( \frac{\partial \hat L}{\partial \phi_{;\mu\nu}}\right)-\frac{\partial \hat{L}}{\partial \phi_{;\gamma\nu}}\Gamma_{\gamma\nu}^{\mu}\right)\nonumber\\
&-\sum_{\substack{\alpha\leq\beta}}g_{\alpha\beta,\mu}\Bigg[\frac{\partial \hat L}{\partial g_{\alpha\beta,\mu}} - 
\sum_{\nu=0}^{3}\frac{1}{n(\mu\nu)}
X_\nu\left( \frac{\partial \hat L}{\partial g_{\alpha\beta,\mu\nu}}\right)\nonumber\\
&-\frac{1}{2}\left(\frac{\partial \hat{L}}{\partial \phi_{;\mu\alpha}}\phi_{;\gamma}g^{\gamma\beta}+\frac{\partial \hat{L}}{\partial \phi_{;\beta\mu}}\phi_{;\gamma}g^{\gamma\alpha}-\frac{\partial \hat{L}}{\partial \phi_{;\alpha\beta}}\phi_{;\gamma}g^{\gamma\mu}\right)\Bigg].
\end{align*}

\end{proposition}

The maps $\Leg_{\mathfrak V}$ and $\widetilde{\Leg}_{\mathfrak V}$ are the \textsl{restricted} and
the {\sl extended Legendre maps} (associated with the Lagrangian density
$\Lag_{\mathfrak V}$), and they satisfy that 
$\mathcal{FL}_{\mathfrak V} = \mu \circ \widetilde{\mathcal{FL}}_{\mathfrak V}$.
For every $j^3_x\phi \in J^3\pi$, we have that
$\operatorname{rank}(\widetilde{\mathcal{FL}}_{\mathfrak V}(j^3_x\phi)) = \operatorname{rank}(\mathcal{FL}_{\mathfrak V}(j^3_x\phi))$.

Remember that, according to \cite{art:Saunders_Crampin90}, 
a second-order Lagrangian density $\mathcal{L} \in \Omega^4(J^2\pi)$ is \textsl{regular} if
$$
\operatorname{rank}(\widetilde{\mathcal{FL}}(j^3_x\phi)) 
= \operatorname{rank}(\mathcal{FL}(j^3_x\phi)) = \dim J^{2}\pi + \dim J^{1}\pi - \dim E = \dim J^2\pi^\ddagger \, ,
$$
otherwise, the Lagrangian density is \textsl{singular}.
Regularity is equivalent to demand that $\mathcal{FL} \colon J^3\pi \to J^{2}\pi^\ddagger$ 
is a submersion onto $J^{2}\pi^\ddagger$ and
this implies that there exist local sections of 
$\mathcal{FL}$. If $\mathcal{FL}$ admits a global section
$\Upsilon \colon J^2\pi^\ddagger \to J^3\pi$, 
then the Lagrangian density is said to be {\sl hyperregular}.
Recall that the regularity of $\mathcal{L}$ determines if the section
$\psi \in \Gamma(\rho_M^r)$ solution to the equation \eqref{eqn:UnifFieldEqSect} lies in
$\mathcal{W}_\mathcal{L}$ or in a submanifold
$\mathcal{W}_f \hookrightarrow \mathcal{W}_\mathcal{L}$ where the section $\psi$ takes values.
In order to obtain this {\sl final constraint submanifold}, the best way is to work with the equation \eqref{eqn:UnifDynEqMultiVF} instead of (\ref{eqn:UnifFieldEqSect}).

\subsection{Field equations for multivector fields}

For a generic cubic Horndenski theory we have, at least, the following constraints:

\begin{theorem}
\label{theo:submanifold} Consider the Horndeski Lagrangian with $G_4(\phi,X)=0$ and $G_5(\phi,X)=0$, 
and let ${\cal W}_f\hookrightarrow\W_r$ be the submanifold defined locally by the constraints
$$
p_g^{\alpha\beta,\mu\nu}-\frac{\partial \hat L}{\partial g_{\alpha\beta,\mu\nu}}=0\quad , \quad 
p_g^{\alpha\beta,\mu}-\hat{L}_g^{\alpha\beta,\mu}=0\quad , \quad 
 \hat{L}_g^{\alpha\beta}=0 \quad , \quad 
X_\tau\hat{L}_g^{\alpha\beta}=0 \ ,
$$
$$
p_{\phi}^{,\mu\nu}-\frac{\partial \hat L}{\partial \phi_{;\mu\nu}}=0\quad , \quad 
p_{\phi}^{,\mu}-\hat{L}_{\phi}^{,\mu}=0\quad , \quad 
 \hat{L}_{\phi}=0 \quad , \quad 
X_\tau\hat{L}_{\phi}=0 \ ;
$$

for $0\leq\alpha\leq\beta\leq3$, $0\leq\mu\leq\nu\leq3$ and $0\leq\tau\leq3$. 
Then, there exist classes of  semiholonomic multivector fields 
$\{{\bf X}\}\subset\mathfrak{X}^4({\cal W}_r)$
which are tangent to $\W_f$ and such that
\beq{}\label{eqn:canonicalisomorphism}
\inn{({\bf X})}\Omega_r|_{\W_f}=0 \quad ,\quad\forall {\bf X}\in\{{\bf X}\}\subset\mathfrak{X}^4({\cal W}_r) \ .
\eeq
\end{theorem}
\begin{proof}
In order to find the final submanifold $\W_f$ we use a local
coordinate procedure which is equivalent to the
constraint algorithm for premultisymplectic field theories.
Bearing in mind \eqref{locholmv}, 
the local expression of a representative of a class of a 
semiholonomic multivector fields, not necessarily integrable, is, in this case, \ref{locholmv}

then, equation \eqref{eqn:UnifDynEqMultiVF} leads to
\bea\label{eqn:UniVec1}
G^{\alpha\beta,\mu}_{g\;\mu}-\frac{\partial\hat{L}}{\partial g_{\alpha\beta}}-\sum_{\mu\leq\nu}p_{\phi}^{,\mu\nu}\phi_{;\gamma}g^{\gamma\alpha}\Gamma_{\mu\nu}^{\beta}=0 \ ,
\\ \label{eqn:UniVec2}
G^{,\mu}_{\phi\;\mu}-\frac{\partial \hat{L}}{\partial \phi}=0 \ ,
\\ \label{eqn:UniVec3}
\sum_{\nu=0}^3 \frac{1}{n(\mu\nu)}G^{\alpha\beta,\mu\nu}_{g\;\nu}-\frac{\partial\hat{L}}{\partial \tilde{g}_{\alpha\beta,\mu}}+p_g^{\alpha\beta,\mu}+\sum_{\rho\leq\sigma}\phi_{;\gamma}p^{,\rho\sigma}_\phi\frac{\partial \Gamma^\gamma_{\rho\sigma}}{\partial \tilde{g}_{\alpha\beta,\mu}}=0 \ ,
\\ \label{eqn:UniVec4}
\sum_{\nu=0}^3 \frac{1}{n(\mu\nu)}G^{,\mu\nu}_{\phi\;\nu}-\frac{\partial\hat{L}}{\partial\phi_{;\mu}}+p_{\phi}^{,\gamma\nu}\Gamma_{\gamma\nu}^{\mu}+p_{\phi}^{,\mu}=0 \ ,
\\ \label{eqn:UniVec5}
p_{g}^{\alpha\beta,\mu\nu}-\frac{\partial \hat{L}}{\partial \tilde{g}_{\alpha\beta,\mu\nu}}=0 \ ,
\\ \label{eqn:UniVec6}
p_{\phi}^{,\mu\nu}-\frac{\partial \hat{L}}{\partial \phi_{;\mu\nu}}=0 \ .
\eea
Expressions \eqref{eqn:UniVec5} and \eqref{eqn:UniVec6} are constraints that define the compatibility submanifold $\mathcal{W}_c\hookrightarrow\W_r$. 
If we require tangency of the multivector field to $\mathcal{W}_c$, 
$$
\Lie(X_\tau)\left(p_{g}^{\alpha\beta,\mu\nu}-\frac{\partial \hat{L}}{\partial g_{\alpha\beta,\mu\nu}}\right)\vert_{\W_c}=0 \ ,
$$
$$
\Lie(X_\tau)\left(p_{\phi}^{,\mu\nu}-\frac{\partial \hat{L}}{\partial \phi_{;\mu\nu}}\right)\vert_{\W_c}=0 \ ,
$$

we get

\beq
G^{\alpha\beta,\mu\nu}_{g\;\tau}=X_{\tau}\left(\frac{\partial\hat{L}}{\partial \tilde{g}_{\alpha\beta,\mu\nu}}\right)
\quad ; \quad \mbox{\rm (on $\W_c$)} \ ,
\label{G1}
\eeq

\beq
G^{,\mu\nu}_{\phi\;\tau}=X_{\tau}\left(\frac{\partial\hat{L}}{\partial \phi_{;\mu\nu}}\right)
\quad ; \quad \mbox{\rm (on $\W_c$)} \ .
\label{G2}
\eeq

Contracting $\tau$ and $\nu$ and combining these expressions with \eqref{eqn:UniVec3} and \eqref{eqn:UniVec4} leads to

\beq
0=\sum_{\nu=0}^3\frac{1}{n(\mu\nu)}X_{\nu}\left(\frac{\partial \hat{L}}{\partial \tilde{g}_{\alpha\beta,\mu\nu}}\right)-\frac{\partial\hat{L}}{\partial \tilde{g}_{\alpha\beta,\mu}}+p_g^{\alpha\beta,\mu}+\sum_{\rho\leq\sigma}\phi_{;\gamma}p^{,\rho\sigma}_\phi\frac{\partial \Gamma^\gamma_{\rho\sigma}}{\partial \tilde{g}_{\alpha\beta,\mu}}
\quad ; \quad \mbox{\rm (on $\W_c$)} \ ,
\label{Wlag1}
\eeq

\beq
0=\sum_{\nu=0}^3\frac{1}{n(\mu\nu)}X_{\nu}\left(\frac{\partial\hat{L}}{\partial \phi_{;\mu\nu}}\right)-\frac{\partial \hat{L}}{\partial \phi_{;\mu}}+p_{\phi}^{,\gamma\nu}\Gamma_{\gamma\nu}^{\mu}+p_{\phi}^{,\mu}
\quad ; \quad \mbox{\rm (on $\W_c$)} \ ,
\label{Wlag2}
\eeq

which can be rewritten as $
p_g^{\alpha\beta,\mu}=\hat{L}^{\alpha\beta,\mu}_g\,;\quad p_{\phi}^{,\mu}=\hat{L}_\phi^{\mu}\,$, respectively. These constraints define the submanifold $\mathcal{W}_\mathcal{L}\hookrightarrow\mathcal{W}_c$.
Imposing tangency conditions on these, gives

$$
G^{\alpha\beta,\mu}_{g\;\tau}=X_\tau\left(\hat{L}^{\alpha\beta,\mu}_g\right)\,;\quad G^{,\mu}_{\phi\;\tau}=X_\tau\left(\hat{L}_\phi^{\mu}\right)\,;\quad \mbox{\rm (on $\W_{\Lag}$)} 
$$

If we contract $\mu$ and $\tau$, and use \eqref{eqn:UniVec1}, \eqref{eqn:UniVec2}, \eqref{eqn:UniVec5}, \eqref{eqn:UniVec6}, \eqref{G1} and \eqref{G2}, we get

\begin{align}\label{eleqns1}
    0=&\frac{\partial \hat{L}}{\partial g_{\alpha\beta}}-X_{\mu}\left(\hat{L}^{\alpha\beta,\mu}_g\right)+\sum_{\mu\leq\nu}p_{\phi}^{,\mu\nu}\phi_{;\gamma}g^{\gamma\alpha}\Gamma_{\mu\nu}^{\beta}    \quad ; \quad \mbox{\rm (on $\W_{\Lag}$)} \ ,
\end{align}
\begin{align}\label{eleqns2}
    0=&\frac{\partial \hat{L}}{\partial \phi}-X_\tau\left(\hat{L}_\phi^{\mu}\right)    \quad ; \quad \mbox{\rm (on $\W_{\Lag}$)} \ .
\end{align}

These results hold for the full Horndeski's theory since up to this point we have only assumed that our Lagrangian is constructed out of the metric tensor, a scalar field and its first and second order derivatives. To provide a better insight of the physical meaning of expressions (\ref{eleqns1}) and (\ref{eleqns2}) we shall consider the Horndeski Lagrangian with $G_4(\phi,X)=0$ and $G_5(\phi,X)=0$ and after carefully computing these expressions explicitly, we get

\begin{align}
    \hat{L}_\phi^{,\mu}=&-\sqrt{-g}\Bigg\{g^{\mu\nu}\phi_{;\nu}\left[1+\frac{\partial G_2{\phi,X}}{\partial X}+\frac{\partial G_3(\phi,X)}{\partial \phi}+\frac{\partial G_3(\phi,X)}{\partial X}\square\phi\right]\nonumber\\
    &+\sum_{\nu=0}^{3}\frac{1}{n(\mu\nu)}g^{\mu\nu}g^{\alpha\beta}\phi_{;\alpha}\left(\phi_{;\beta\nu}+\phi_{;\gamma}\Gamma^{\gamma}_{\beta\nu}\right)\frac{\partial G_3(\phi,X)}{\partial X}\Bigg\}
\end{align}

\begin{align}
    \hat{L}_g^{\alpha\beta,\mu}=\frac{1}{2}\sqrt{-g}\Big\{g_{\rho \lambda,\sigma}\Big[&-3g^{\rho \lambda}g^{\mu(\alpha}g^{\beta) \lambda}+2g^{\rho \lambda}g^{\mu(\alpha}g^{\beta)\sigma}+2g^{\alpha\beta}g^{\rho \sigma}g^{\mu \lambda}\nonumber\\
    &+3g^{\mu \sigma}g^{\lambda(\alpha}g^{\beta)\rho}-2g^{\rho\mu}g^{\lambda(\alpha}g^{\beta)\sigma}-g^{\alpha\beta}g^{\rho \lambda}g^{\mu \sigma}\nonumber\\
    &-\frac{1}{2}g^{\rho b}g^{\mu \sigma}g^{\alpha\beta}+g^{\mu(\rho}g^{\lambda)c}g^{\alpha\beta}+g^{\mu \sigma}g^{\alpha(\rho}g^{\lambda)\beta}\Big]\nonumber\\
    &-G_3(\phi,X)\phi_{;\gamma}(g^{\gamma\beta}g^{\alpha\mu}+g^{\gamma\alpha}g^{\beta\mu}-g^{\gamma\mu}g^{\alpha\beta})\Big\}\,.
\end{align}

Therefore:

\begin{align}\label{eqn:ELG}
\hat L^{\alpha\beta}_g=&
-\sqrt{g}n(\alpha\beta)\Bigg[R^{\alpha\beta}-\frac{1}{2}\left(R+X+G_2+2X\frac{\partial G_3}{\partial \phi}\right)g^{\alpha\beta}-\frac{1}{2}\left(1+\frac{\partial G_2}{\partial X}+2\frac{\partial G_3}{\partial \phi}\right)g^{\alpha\rho}g^{\beta\sigma}\phi_{;\rho}\phi_{;\sigma}\nonumber\\
&+\frac{1}{2}\frac{\partial G_3}{\partial X}\left(-g^{\rho\sigma}g^{\alpha\delta}g^{\beta\sigma}\phi_{;\rho\sigma}\phi_{;\delta}\phi_{;\sigma}+2g^{\alpha\rho}g^{\beta\sigma}g^{\gamma\delta}\phi_{;(\rho}\phi_{;|\gamma|\sigma)}\phi_{;\delta}\right)\Bigg]=0\quad ; \quad \mbox{\rm (on $\W_{\Lag_{\mathfrak V}}$)} \ ,
\end{align}

\begin{align}\label{eqn:ELPHI}
\hat L_{\phi}=&-\sqrt{-g}n(\alpha\beta)\Bigg[-\left(1+\frac{\partial G_2}{\partial X}+2X\frac{\partial^2 G_2}{\partial X^2}+2\frac{\partial G_3}{\partial \phi}+2X\frac{\partial^2 G_3}{\partial X\partial \phi}\right)g^{\mu\nu}\phi_{;\mu\nu}\nonumber\\
&+g^{\mu\rho}g^{\nu\sigma}R_{\mu\nu}\frac{\partial G_3}{\partial X}\phi_{;\rho}\phi_{;\sigma}-\left(\frac{\partial^2 G_2}{\partial X^2}+2\frac{\partial^2 G_3}{\partial X\partial \phi}\right)\left(g^{\mu\nu}\phi_{;\mu\nu}g^{\rho\sigma}\phi_{;\rho}\phi_{;\sigma}-g^{\rho\nu}g^{\sigma\mu}\phi_{;\rho}\phi_{;\sigma}\phi_{;\nu\mu}\right)\nonumber\\
&-\frac{\partial G_3}{\partial X}\left(g^{\mu\nu}\phi_{;\mu\nu}g^{\rho\sigma}\phi_{;\rho\sigma}-g^{\mu\rho}g^{\nu\sigma}\phi_{;\mu\nu}\phi_{;\rho\sigma}\right)+\left(2X\frac{\partial^2 G_3}{\partial \phi^2}+\frac{\partial^2 G_2}{\partial X\partial \phi}\right)-\frac{\partial G_2}{\partial \phi}\nonumber\\
&+\frac{\partial^2 G_3}{\partial X^2}g^{\mu\sigma}g^{\nu\rho}\phi_{;\mu}\phi_{;\nu}\left(g^{\gamma\delta}\phi_{;\gamma\delta}\phi_{;\rho\sigma}-g^{\gamma\delta}\phi_{;\gamma\sigma}\phi_{;\rho\delta}\right)\Bigg]=0\quad,
\end{align}

Expresions (\ref{eqn:ELG}) and (\ref{eqn:ELPHI}) are the Euler-Lagrange equations, and  when they are evaluated on sections in $\W_{\Lag_{\mathfrak V}}$ we recover  the cubic Horndeski equations of motion.

It is important to remark that $\hat L_g^{\alpha\beta}$ and $\hat L_{\phi}$ do not depend on any of the momenta. There is a dependence on the velocities and accelerations of both the metric and the scalar field, but not on higher-order velocities of any of them. Hence, $\hat L_g^{\alpha\beta}$ and $\hat L_{\phi}$ project onto ${J^2\pi}$ and they can be regarded as new constraints defining locally a submanifold $\mathcal{W}_1\hookrightarrow\mathcal{W}_{\mathcal{L}_{\mathfrak V}}\hookrightarrow\mathcal{W}_r$. Once again, demanding tangency of the multivector field to this new manifold we get

\begin{align*}
\Lie(X_\tau)\hat{L}_g^{\alpha\beta}\vert_{\W_1}=0 \ ,\\
\Lie(X_\tau)\hat{L}_{\phi}\vert_{\W_1}=0 \ ,
\end{align*}

which are just

\begin{align*}
X_\tau\left(\hat{L}_g^{\alpha\beta}\right)&=0 ; \quad \mbox{\rm (on $\W_1$)}\ ,\\
X_\tau\left(\hat{L}_{\phi}\right)&=0; \quad \mbox{\rm (on $\W_1$)}\ .
\end{align*}

These are new constraints again that project onto $J^3\pi$. They define locally the submanifold 
$\mathcal{W}_f\hookrightarrow\mathcal{W}_1\hookrightarrow\mathcal{W}_\mathcal{L}\hookrightarrow\mathcal{W}_r$. This manifold $\mathcal{W}_f$ is the final constraint submanifold because there exist
holonomic multivector fields, solutions to (\ref{eqn:canonicalisomorphism}).
Finally, the new tangency conditions, 

\begin{align*}
\Lie(X_\sigma)X_{\tau}\left(\hat{L}_g^{\alpha\beta}\right)\vert_{\W_1}=0 \ ,\\
\Lie(X_\sigma)X_{\tau}\left(\hat{L}_{\phi}\right)\vert_{\W_1}=0 \ ,
\end{align*}

which are explicitly

\begin{align*}
X_\sigma\left(X_{\tau}\left(\hat{L}_g^{\alpha\beta}\right)\right)=0 ; \quad \mbox{\rm (on $\W_1$)}\ ,\\
X_\sigma\left(X_{\tau}\left(\hat{L}_{\phi}\right)\right)=0 \quad \mbox{\rm (on $\W_1$)}\ .
\end{align*}

These allow us to determine the remaining components of (\ref{locholmv}),
$F_{g\;\alpha\beta,\mu\nu\tau\lambda}$ and $F_{\phi\;,\mu\nu\tau\lambda}$. Finally, the complete set of constraints that define
the final constraint submanifold $\W_f\hookrightarrow\W_r$ 
are 

$$
p_g^{\alpha\beta,\mu\nu}-\frac{\partial \hat L}{\partial g_{\alpha\beta,\mu\nu}}=0\quad , \quad 
p_g^{\alpha\beta,\mu}-\hat{L}_g^{\alpha\beta,\mu}=0\quad , \quad 
 \hat{L}_g^{\alpha\beta}=0 \quad , \quad 
X_\tau\left(\hat{L}_g^{\alpha\beta}\right)=0 \ ,
$$
$$
p_{\phi}^{,\mu\nu}-\frac{\partial \hat L}{\partial \phi_{;\mu\nu}}=0\quad , \quad 
p_{\phi}^{,\mu}-\hat{L}_{\phi}^{,\mu}=0\quad , \quad 
 \hat{L}_{\phi}=0 \quad , \quad 
X_\tau\left(\hat{L}_{\phi}\right)=0 \ ;
$$
\end{proof}

In contrasts with the Hilbert-Einstein case, we cannot assume that there exists a holonomic solution (that is, integrable) in $\W_f$. Depending on the particular choice of $G_1$, $G_2$ and $G_3$, new constraints may appear when demanding integrability of the multivector field in the constraint algorithm.

\subsection{Field equations for sections}

Provided that we now know the solution for the holonomic multivector fields, we can evaluate equation (\ref{pdesect}) to recover the field equations for sections. 

\begin{align}
\derpar{\psi_{g\;\alpha\beta}}{x^\mu}&=\psi_{g\;\alpha\beta,\mu} \ , \ \label{eqn:holonomy1}\\
\derpar{\psi_{g\;\alpha\beta,\mu}}{x^\nu}&=\psi_{g\;\alpha\beta,\mu\nu} \ , \ \label{eqn:holonomy2}\\
\nabla_{\mu}\psi_{\phi}&=\psi_{\phi\; ;\mu} \ , \ \label{eqn:holonomy3}\\
\nabla_{\nu}\psi_{\phi\;;\mu}&=\psi_{\phi\; ;\mu\nu} \ , \ \label{eqn:holonomy4}\\
\derpar{\psi_g^{\alpha\beta,\mu}}{x^\mu}&=\derpar{\hat L}{g_{\alpha\beta}}+\psi_{\phi}^{,\mu\nu}\psi_{\phi\; ;\gamma}\psi_{g}^{\gamma\alpha}\Gamma^{\beta}_{\mu\nu}  \ , \ \\
\derpar{\psi_{\phi}^{,\mu}}{x^\mu}&=\derpar{\hat L}{\phi}  \ , \\
\derpar{\psi_g^{g\;\alpha\beta,\mu\nu}}{x^\mu}&=\derpar{\hat L}{g_{\alpha\beta,\mu}}-\psi_{g}^{\alpha\beta,\mu}-\frac{1}{2}\psi_{\phi\; ;\gamma}\left(\psi_{\phi}^{,\mu\alpha}\psi_{g}^{\gamma\beta}+\psi_{\phi}^{,\beta\mu}\psi_{g}^{\gamma\alpha}-\psi_{\phi}^{,\alpha\beta}\psi_{g}^{\gamma\mu}\right) \label{eqn:transformacionlegendresecciones1}\\
\derpar{\psi_{\phi}^{,\mu\nu}}{x^\nu}&=\derpar{\hat L}{\phi_{;\mu}}-\psi_{\phi}^{,\gamma\nu}\Gamma^{\mu}_{\gamma\nu}-\psi_{\phi}^{,\mu}  \ ,\\
\psi_{g}^{\alpha\beta,\mu\nu}&=\derpar{\hat L}{g_{\alpha\beta,\mu\nu}}\ ,\\
\psi_{\phi}^{,\mu\nu}&=\derpar{\hat L}{\phi_{;\mu\nu}}\ \label{eqn:transformacionlegendresecciones4}.
\end{align}

Equations (\ref{eqn:holonomy1})-(\ref{eqn:holonomy4}) are the holonomy conditions for the multivector field. Equations (\ref{eqn:transformacionlegendresecciones1})-(\ref{eqn:transformacionlegendresecciones4}) define the Legendre transformations. 


\section{Hamiltonian formalism}
\label{Section5}

The covariant Hamiltonian formalism takes place in the image of the Legendre Transformation. For singular Lagrangian this space could be highly degenerate. The Legendre maps in our case are given by Proposition \ref{prop:GraphLegMapSect}.
Then
$$
\Tan_{{j^3_x}\phi}\mathcal{FL}_{\mathfrak V}=\left( \begin{array}{ccccccccc}
1 & 0 & 0 & 0 & 0 & 0 & 0 & 0 & 0\\
0 & 1 & 0 & 0 & 0 & 0 & 0 & 0 & 0\\
0 & 0 & 1 & 0 & 0 & 0 & 0 & 0 & 0\\
0 & 0 & 0 & 1 & 0 & 0 & 0 & 0 & 0\\
0 & 0 & 0 & 0 & 1 & 0 & 0 & 0 & 0\\
0 & \displaystyle\frac{\partial \hat L_g^{\alpha\beta,\mu}}{\partial g_{\gamma\delta}} & \displaystyle\frac{\partial \hat L_g^{\alpha\beta,\mu}}{\partial \phi} & \displaystyle\frac{\partial \hat L_g^{\alpha\beta,\mu}}{\partial g_{\gamma\delta,\tau}} & \displaystyle\frac{\partial \hat L_g^{\alpha\beta,\mu}}{\partial \phi_{;\tau}} & 0 & 0 & 0 & 0\\
0 & \displaystyle\frac{\partial \hat L_{\phi}^{,\mu}}{\partial g_{\gamma\delta}} & \displaystyle\frac{\partial \hat L_{\phi}^{,\mu}}{\partial \phi} & \displaystyle\frac{\partial \hat L_{\phi}^{,\mu}}{\partial g_{\gamma\delta,\tau}} & \displaystyle\frac{\partial \hat L_{\phi}^{,\mu}}{\partial \phi_{;\tau}} & \displaystyle\frac{\partial \hat L_{\phi}^{,\mu}}{\partial \phi_{;\tau\lambda}} & 0 & 0 & 0\\
0 & \displaystyle\frac{\partial^2\hat L}{\partial g_{\gamma\delta}\partial g_{\alpha\beta,\mu\nu}} & \displaystyle\frac{\partial^2\hat L}{\partial \phi\partial g_{\alpha\beta,\mu\nu}} & \displaystyle\frac{\partial^2\hat L}{\partial g_{\gamma\delta,\tau}\partial g_{\alpha\beta,\mu\nu}} & \displaystyle\frac{\partial^2\hat L}{\partial \phi_{;\tau}\partial g_{\alpha\beta,\mu\nu}} & 0 & 0 & 0 & 0\\
0 & \displaystyle\frac{\partial^2\hat L}{\partial g_{\gamma\delta}\partial \phi_{;\mu\nu}} & \displaystyle\frac{\partial^2\hat L}{\partial \phi\partial \phi_{;\mu\nu}} & \displaystyle\frac{\partial^2\hat L}{\partial g_{\gamma\delta,\tau}\partial \phi_{;\mu\nu}} & \displaystyle\frac{\partial^2\hat L}{\partial \phi_{;\tau}\partial \phi_{;\mu\nu}} & 0 & 0 & 0 & 0
\end{array} \right) \ 
$$

Notice that, in general, ${\rm rank}(\Tan_{{j^3_x}\phi}\mathcal{FL}_{\mathfrak V})\geq59$, depending on the arbitrary function $G_3(\phi,X)$. Also, locally
\beq
\ker\,(\Leg_{{\mathfrak V}})_*=
\ker\,\Omega_{\Lag_{\mathfrak V}}\supset{\left<\frac{\partial}{\partial g_{\alpha\beta,\mu\nu}},\frac{\partial}{\partial g_{\alpha\beta,\mu\nu\lambda}},\frac{\partial}{\partial\phi_{;\mu\nu\lambda}}\right>}_{0\leq\alpha\leq\beta\leq3;\, 0\leq\mu\leq\nu\leq\lambda\leq3}\ ,
\label{kerfl}
\eeq
hence $\mathcal{FL}_{\mathfrak V}$ is highly degerated. 
 
We denote $\widetilde{\mathcal{P}}=\widetilde{\mathcal{FL}}_{\mathfrak V}(J^3\pi) \stackrel{\tilde{\jmath}}{\hookrightarrow} J^2\pi^\dagger$
and $\mathcal{P}=\mathcal{FL}_{\mathfrak V}(J^3\pi)\stackrel{\jmath}{\hookrightarrow} J^2\pi^\ddagger$,
and let $\mathcal{FL}_{\mathfrak V}^o$ be the map defined by 
$\mathcal{FL}_{\mathfrak V}=\jmath\circ \mathcal{FL}_{\mathfrak V}^o$ and 
$\bar{\pi}_{\mathcal{P}}\colon\mathcal{P}\to M$ 
the natural projection.
In order to assure the existence of the Hamiltonian formalism 
it is needed that the Lagrangian density $\mathcal{L}_{\mathfrak V}\in\Omega^4(J^2\pi)$ is, at least, 
{\sl almost-regular}; i.e, $\mathcal{P}$ is a closed submanifold
of $J^2\pi^\ddagger$, $\mathcal{FL}_{\mathfrak V}$ is a submersion onto its image
and, for every $j^3_x\phi \in J^3\pi$, the fibers 
$\mathcal{FL}_{\mathfrak V}^{-1}(\mathcal{FL}_{\mathfrak V}(j^3_x\phi))$
are connected submanifolds of $J^3\pi$. For more details in almost-regular Lagrangians and how to recover the Hamiltonian formalism from the unified Lagrangian-Hamiltonian formalism, we recommend consulting references \cite{pere1,art:Roman09}.

The proof that the Hilbert-Einstein Lagrangian is almost-regular is based on the fact that $\mathcal{P}$ is diffeomorphic to the first jet of the corresponding fiber bundle\cite{GasetGR}. This property is closely related to the fact that the Euler-Lagrange equations (Einstein's Field Equations) are second-order, although one expects fourth-order equations for a second-order Lagrangian. This topic is called order-reduction (or projectability) of a theory \cite{Rosado,rosado2,gaset_order_2016}. Horndenski Lagrangians are constructed such that the corresponding field equations are second-order, therefore, one hopes to proceed in a similar way as in the Hilbert-Einstein case. Nevertheless, the fact that the Euler-Lagrangians equations projects to lower order doesn't implies that the geometric structures also project to a lower order. The Horndenski theories that have this property are characterised by proposition \ref{prop:proyecta}. 

A form $\alpha\in \Omega^*(J^3\pi)$ projects to $J^s\pi$, $s=1,2$, if it is $\pi^3_s$-basic, that is, $\Lie_Y\alpha=0$ for all vector fields $Y$ vertical with respect to $\pi^3_s$. The multisymplectic Lagrangian system is $(J^3\pi,\Omega_{\mathcal{L}})$, where $\Omega_{\mathcal{L}_{\mathcal{V}}}= \widetilde{\mathcal{FL}_{\mathcal{V}}}^{\ *}\Omega_1^s$ is the Poincaré-Cartan form.

\begin{proposition}\label{prop:proyecta}
The Poincaré-Cartan form $\Omega_{\mathcal{L}_{\mathcal{V}}}$ of a cubic Horndeski Lagrangian projects to $J^1\pi$ if, and only if,
$$\frac{\partial G_3(\phi,X)}{\partial X}=0\,.$$
\end{proposition}
\begin{proof}

The necessary and sufficient conditions for the associated Poincaré-Cartan form of a second-order theory to project on $J^1\pi$, according to \cite{Rosado} and \cite{rosado2}, are that $L\in C^\infty(J^2\pi)$ is an affine function with respect to the affine structure of $p_1^2:J^2\pi\rightarrow J^1\pi$, i.e.,

\begin{equation}
    L=L^{ij}_\alpha y^{\alpha}_{(ij)}+L_0, L^{ij}_\alpha=L^{ji}_\alpha\in C^{\infty}(J^1\pi), L_0\in C^{\infty}(J^1\pi),
\end{equation}

and the following equations hold:

\begin{equation}\label{condicionesrosado}
    2\frac{\partial L_{\beta}^{hi}}{\partial y_{a}^{\alpha}}-\frac{\partial L_{\alpha}^{ai}}{\partial y_{h}^{\beta}}-\frac{\partial L_{\alpha}^{ah}}{\partial y_{i}^{\beta}}=0, a,h,i=1,...n, \alpha,\beta=1,...,m,
\end{equation}

which in cubic Horndeski's theory translate into 

\begin{equation}
    L=\sum_{\alpha\leq\beta}L_g^{\alpha\beta,\mu\nu}g_{\alpha\beta,\mu\nu}+L_\phi^{,\mu\nu}\phi_{,\mu\nu}+L_0,
\end{equation}

where

\begin{align}
    L_0=&\sqrt{-g}\Bigg\{g^{\alpha\beta}\Big[g^{\gamma\delta}(g_{\delta\mu,\beta}\Gamma^{\mu}_{\alpha\gamma}-g_{\delta\mu,\gamma}\Gamma^{\mu}_{\alpha\beta})+\Gamma^{\delta}_{\alpha\beta}\Gamma^{\gamma}_{\gamma\delta}-\Gamma^{\delta}_{\alpha\gamma}\Gamma^{\gamma}_{\beta\delta}\Big]\nonumber\\
    &+X+G_2(\phi,X)-\phi_{,\gamma}\Gamma^{\gamma}_{\nu\mu}g^{\mu\nu}G_3(\phi,X)\Bigg\},
\end{align}

\begin{equation}
    L_g^{\alpha\beta,\mu\nu}=\frac{n(\alpha\beta)}{2}\sqrt{-g}\left(g^{\alpha\mu}g^{\beta\nu}+g^{\alpha\nu}g^{\beta\mu}-2g^{\alpha\beta}g^{\mu\nu}\right),
\end{equation}

\begin{equation}
    L_{\phi}^{,\mu\nu}=\sqrt{-g}g^{\mu\nu}G_3(\phi,X),
\end{equation}

and the equations (\ref{condicionesrosado}) hold if, and only, if

\begin{equation}\label{Constraintrosado}
    \sqrt{-g}\phi_{,\delta}\frac{\partial G_3(\phi,X)}{\partial X}\left(-2g^{\mu\nu}g^{\gamma\delta}+g^{\gamma\nu}g^{\mu\delta}+g^{\gamma\mu}g^{\nu\delta}\right)=0.
\end{equation}

Equation \ref{Constraintrosado} only holds if $\frac{\partial G_3(\phi,X)}{\partial X}=0$. 
\end{proof}

If $\frac{\partial G_3}{\partial X}=0$, then we should expect that the systems behaves like a first-order system. We will study this particular case first, and then we present the general case.

\subsection{Hamiltonian formalism for a particular case}

Throughout this subsection we shall consider that $\frac{\partial G_3(\phi,X)}{\partial X}=0$, i.e. $G_3(\phi,X)=G_3(\phi)$. Hence

\begin{equation}
L_{\mathfrak V}=\frac{1}{16\pi G}\sqrt{|g|}\left[R+X+G_2(\phi,X)+G_3(\phi)\square\phi\right].
\end{equation}

\begin{proposition}
\label{prop:AlmReg}
$L_{\mathfrak V}$ is an almost-regular Lagrangian
and
$\mathcal{P}$ is diffeomorphic to $J^1\pi$.
\label{1stdifeo}
\end{proposition}
\begin{proof}
$\mathcal{P}$ is a closed submanifold of $J^2\pi^\ddagger$ 
since it is defined by the constraints
\begin{align*}
p_g^{\alpha\beta,\mu\nu}-\frac{\partial \hat L}{\partial g_{\alpha\beta,\mu\nu}}&=0;\quad 
p_g^{\alpha\beta,\mu}-\hat L_g^{\alpha\beta,\mu}=0\ ,\\
p_{\phi}^{,\mu\nu}-\frac{\partial \hat L}{\partial \phi_{;\mu\nu}}&=0;\quad \quad \quad 
p_{\phi}^{,\mu}-\hat L_{\phi}^{,\mu}=0\ .
\end{align*}

The dimension of $\mathcal{P}$ is $4+10+1+40+4=59$ and,
as ${\rm rank}(\Tan\mathcal{FL}_{\mathfrak V})=59$ in every point, 
$\Tan\mathcal{FL}_{\mathfrak V}$ is surjective and $\mathcal{FL}_{\mathfrak V}$ is a submersion. Moreover,
\beq
\ker\,(\Leg_{{\mathfrak V}})_*={\left<\frac{\partial}{\partial g_{\alpha\beta,\mu\nu}},\frac{\partial}{\partial g_{\alpha\beta,\mu\nu\lambda}},\frac{\partial}{\partial\phi_{;\mu\nu\lambda}}\right>}_{0\leq\alpha\leq\beta\leq3;\, 0\leq\mu\leq\nu\leq\lambda\leq3}\ ,
\label{kerfl}
\eeq
therefore, the fibers of the Legendre map
are the fibers of the projection $\bar\pi^3_1$. As we consider metric with fixed signature, they are connected submanifolds of $J^3\pi$. 

Taking any local section $\phi$ of the projection $\pi^3_1$,
the map $\Phi=\Leg_{\mathfrak V}\circ\phi\colon J^1\pi\to{\cal P}$
is a local diffeomorphism and it does not depend on the chosen section.
Therefore, ${\cal P}$ and $J^1\pi$ are diffeomorphic.
$$
\xymatrix{
 \ J^3\pi \ar@/_0pc/[rr]^{\Leg_{\mathfrak V}} \ar@/_0pc/[ddr]^{\pi^3_1} 
\ & \ & \ 
\mathcal{P}\subset J^2\pi^\ddagger 
 \\  \ & \ & \ & \ &  \\
& \ J^1\pi \  \ar@/^1pc/[uul]^{\phi} \ar@/_0pc/[uur]_{\Phi}& 
}
$$
\end{proof}

Then, the $\mu$-transverse submanifolds
$\widetilde{\cal P}$ and ${\cal P}$ are diffeomorphic
and the diffeomorphism, denoted
 $\tilde\mu\colon\tilde{\cal P}\to{\cal P}$,
is just the restriction of the projection $\mu$ to $\tilde{\cal P}$.
Therefore we can define a
{\sl Hamiltonian $\mu$-section} as $h_{\mathfrak V} = \tilde{\jmath} \circ \widetilde{\mu}^{-1}$, 
which  is specified by a local Hamiltonian function
$H_{\mathcal{P}} \in C^\infty(\mathcal{P})$; that is,
$$h_{\mathfrak V}(x^\mu,g_{\alpha\beta},\phi,g_{\alpha\beta,\mu},\phi_{\mu},p_g^{\alpha\beta,\mu},p_{\phi}^{,\mu},p_g^{\alpha\beta,\mu\nu},p_{\phi}^{,\mu\nu})=(x^\mu,g_{\alpha\beta},\phi,g_{\alpha\beta,\mu},\phi_{\mu},-H_{\mathcal{P}},p_g^{\alpha\beta,\mu},p_{\phi}^{,\mu},p_g^{\alpha\beta,\mu\nu},p_{\phi}^{,\mu\nu}).$$

This function $H_{\mathfrak V}$ is the Hamiltonian function defined on 
$\mathcal{P}$ and is given by
$H=(\Leg_{\mathfrak V}^o)^*\,H_{\mathfrak V}$;
where $H$, which is $\Leg_{\mathfrak V}^o$-projectable, is 

\begin{align}
    H=\sum_{\alpha\leq\beta}L_g^{\alpha\beta,\mu\nu}g_{\alpha\beta,\mu\nu}+L_{\phi}^{,\mu\nu}\left(\phi_{;\mu\nu}+\phi_{;\gamma}\Gamma^{\gamma}_{\mu\nu}\right)+\sum_{\alpha\leq\beta}L_g^{\alpha\beta,\mu}g_{\alpha\beta,\mu}+L_{\phi}^{,\mu}\phi_{;\mu}-L,
\end{align}

This is summarised in the following diagram:

$$
\xymatrix{
\widetilde{\mathcal{P}} \ar[rr]^-{\tilde{\jmath}} \ar[d]^-{\tilde{\mu}} & \ & J^{2}\pi^\dagger \ar[d]^-{\mu} & \ & \mathcal{W} \ar[d]_-{\mu_\mathcal{W}} \ar[ll]_-{\rho_2} \\
\mathcal{P} \ar[rr]^-{\jmath} \ar[urr]^-{h_{\mathfrak V}}& \ & J^{2}\pi^\ddagger & \ & \mathcal{W}_r \ar@/_0.7pc/[u]_{\hat{h}} \ar[ll]_-{\rho_2^r}
}
$$

Now, it is possible to define the Hamiltonian forms

$$
\Theta_{h_{\mathfrak V}}:=h_{\mathfrak V}^*\Theta_1^s \in\df^4({\cal P}) \quad , \quad
\Omega_{h_{\mathfrak V}}:=-\d\Theta_{h_{\mathfrak V}}=h_{\mathfrak V}^*\Omega_1^s \in\df^{5}({\cal P}) \ ,
$$
and thus we have the Hamiltonian system $({\cal P},\Omega_{h_{\mathfrak V}})$.
Then, the \textsl{Hamiltonian problem} associated with this system
consists in finding holonomic sections $\psi_h\colon M\rightarrow\mathcal{P}$ 
satisfying any of the following equivalent conditions:
\begin{enumerate}
\item $\psi_h$ is a solution to the equation
\beq{}\label{eq:hsec}
\psi_h^*i(X)\Omega_{h_{\mathfrak V}}=0 \quad, \quad 
\mbox{\rm for every $X \in\mathfrak{X}(\mathcal{P})$}\ .
\eeq
\item $\psi_h$ is an integral section of a multivector field contained 
in a class of holonomic multivector fields 
$\{{\bf X}_h\}\subset\mathfrak{X}^4(\mathcal{P})$
satisfying the equation
\beq{}\label{eq:hvf}
\inn{({\bf X}_h)}\Omega_{h_{\mathfrak V}}=0 \quad ,\quad
\forall {\bf X}_h\in\{{\bf X}_h\}\subset\mathfrak{X}^4(\mathcal{P}) \ .
\eeq
\end{enumerate}
(Here, holonomic sections and multivector fields are defined as in $J^2\pi^\dagger$).
The solutions of the Hamiltonian formalism can be recovered geometrically from the unified formalism using the adequate projections (see \cite{pere2} for more details). Nevertheless, we will continue by presenting the local expression of equations \ref{eq:hvf}.

{\bf Formulation using multimomentum coordinates.}

The natural coordinates of $J^2\pi^\ddagger$ are 
\begin{equation*}
   (x^\mu,g_{\alpha\beta},\phi,g_{\alpha\beta,\mu},\phi_{;\mu},p_g^{\alpha\beta,\mu},p_{\phi}^{,\mu},p_g^{\alpha\beta,\mu\nu},p_{\phi}^{,\mu\nu}), 
\end{equation*}

which contain the multimomenta and the velocities. These are the expected coordinates for a Hamiltonian formulation of a second-order regular Lagrangian. Nevertheless, our Lagrangian is singular and the Hamiltonian formulation takes place in the submanifold $\mathcal{P}$. Since it is diffeomorphic to $J^1\pi$ by proposition \ref{prop:AlmReg}, a natural set of coordinates is
\begin{equation*}
   (x^\mu,g_{\alpha\beta},\phi,g_{\alpha\beta,\mu},\phi_{,\mu})\,.
\end{equation*}
This is an uninteresting coordinate system, as the resulting equations are identical than the Lagrangian ones. It is customary to write the Hamiltonian in terms of the positions and multimomenta only, so we need to isolate the velocities to be able to write the Hamiltonian in these terms. The relation between momenta and velocities is given by

\begin{align}
    p_g^{\alpha\beta,\mu}=&-\frac{1}{2}\sqrt{-g}\Big\{G_3(\phi)\phi_{;\gamma}\left(g^{\gamma\beta}g^{\alpha\mu}+g^{\gamma\alpha}g^{\beta\mu}-g^{\gamma\mu}g^{\alpha\beta}\right)\nonumber\\
    &-g_{\rho\lambda,\sigma}\Big[-3g^{\rho\sigma}g^{\mu(\alpha}g^{\beta) \lambda}+2g^{\rho\lambda}g^{\mu(\alpha}g^{\beta)\sigma}+2g^{\alpha\beta}g^{\rho\sigma}g^{\mu \lambda}\nonumber\\
    &+3g^{\mu \sigma}g^{\lambda(\alpha}g^{\beta)\rho}-2g^{\rho\mu}g^{\lambda(\alpha}g^{\beta)\sigma}-g^{\alpha\beta}g^{\rho\lambda}g^{\mu \sigma}\nonumber\\
    &-\frac{1}{2}g^{\rho\lambda}g^{\mu \sigma}g^{\alpha\beta}+g^{\mu(\rho}g^{\lambda)\sigma}g^{\alpha\beta}+g^{\mu \sigma}g^{\alpha(\rho}g^{\lambda)\beta}\Big]\Big\}\,,
\end{align}

\begin{align}
    p_{\phi}^{,\mu}=&-\sqrt{-g}g^{\mu\nu}\Bigg\{\phi_{;\nu}\left[1+\frac{\partial G_2(\phi,X)}{\partial X}+\frac{\partial G_3(\phi)}{\partial \phi}\right]\Bigg\}
\end{align}

 Unlike in General Relativity, where the $p_g^{\alpha\beta,\mu}$ and $g_{\alpha\beta,\mu}$ are in one to one correspondence \cite{GasetGR}, this is not generally true even in the particular case where $\frac{\partial G_3}{\partial X}=0$. To isolate the velocities we would need to fully specify $G_2(\phi,X)$ and in some cases it would not even be possible to do so.

To illustrate this procedure, we will consider the case $\frac{\partial^2 G_2}{\partial X^2}=0$ and $1+\frac{\partial G_2(\phi,X)}{\partial X}+\frac{\partial G_3(\phi)}{\partial\phi}\neq0$. With this in mind, we isolate the velocities in terms of the positions and multimomenta only:

\begin{align}
    \phi_{;\nu}=&-\frac{1}{\sqrt{-g}}\left[\frac{p_{\phi}^{,\mu}g_{\mu\nu}}{1+\frac{\partial G_2(\phi,X)}{\partial X}+\frac{\partial G_3(\phi)}{\partial\phi}}\right]=U_{,\nu}\\
    g_{\alpha\beta,\mu}=&\frac{1}{3\sqrt{-g}}\frac{1}{n(\alpha\beta)}\Bigg\{p_g^{\lambda\sigma,\nu}-\frac{1}{2}p_{\phi}^{,\delta}g_{\delta\epsilon}\left[\frac{G_3(\phi)}{1+\frac{\partial G_2(\phi,X)}{\partial X}+\frac{\partial G_3(\phi)}{\partial \phi}}\right]\left[g^{\nu\lambda}g^{\epsilon\sigma}+g^{\nu\sigma}g^{\epsilon\lambda}-g^{\lambda\sigma}g^{\epsilon\nu}\right]\Bigg\}\nonumber\\
    &\Bigg\{-2g_{\alpha\lambda}g_{\beta\mu}g_{\sigma\nu}-2g_{\alpha\mu}g_{\beta\lambda}g_{\sigma\nu}+6g_{\alpha\lambda}g_{\beta\sigma}g_{\mu\nu}+g_{\alpha\nu}g_{\beta\mu}g_{\lambda\sigma}+g_{\alpha\mu}g_{\beta\nu}g_{\lambda\sigma}\Bigg\}=V_{\alpha\beta,\mu}
\end{align}

Notice that we require that if $
\frac{\partial G_2(\phi,X)}{\partial X}+\frac{\partial G_3(\phi)}{\partial\phi}=-1\,,
$
 then $p^{,\mu}_\phi=0$ and there is no hope to use $p^{,\mu}_\phi$ as a coordinate instead of $\phi;_\mu$.

Now we can set $(x^\mu,g_{\alpha\beta},\phi,p_g^{\alpha\beta,\mu},p_{\phi}^{,\mu})$ 
as coordinates of $\mathcal{P}$ and then rewrite the Hamiltonian function
$$
H_{\mathfrak V}(x^\mu,g_{\alpha\beta},\phi,p_g^{\alpha\beta,\mu},p_\phi^{,\mu})=
H_{\mathfrak V}(x^\mu,g_{\alpha\beta},\phi,V_{\alpha\beta,\mu}(p_g^{\alpha\beta,\mu},p_\phi^{,\mu},g_{\alpha\beta},\phi),U_{,\mu}(p_\phi^{,\mu},g_{\alpha\beta},\phi)) \ .
$$

The Hamiltonian function is hence
\begin{align*}
    H_{\mathfrak V}=&\sum_{\alpha\leq\beta}p_g^{\alpha\beta,\mu}V_{\alpha\beta,\mu}+\sqrt{-g}G_3(\phi)g^{\mu\nu}U_{;\gamma}\Gamma_{\nu\mu}^{\gamma}+p_{\phi}^{,\mu}U_{,\mu}-\sqrt{-g}\left(X+G_2(\phi,X)\right)\nonumber\\
    &-\sqrt{-g}g^{\rho\lambda}\left[-\frac{1}{2}g^{\sigma \epsilon}g^{\delta \zeta}g_{\epsilon\zeta,\sigma}\left(g_{\lambda\delta,\rho}+g_{\rho \delta,\lambda}-g_{\rho\lambda,\delta}\right)+\frac{1}{2}g^{\sigma \epsilon}g^{\delta\zeta}g_{\epsilon\zeta,\lambda}\left(g_{\rho \delta,\lambda}+g_{\sigma \delta,\rho}-g_{\rho\sigma,\delta}\right)+\Gamma_{\sigma \delta}^\sigma\Gamma_{\lambda\rho}^{\delta}-\Gamma_{\lambda\rho}^{\sigma}\Gamma_{\sigma\rho}^{\delta}\right]
\end{align*}

The field equations are derived again from \eqref{eq:hvf} expressed using the new coordinates.
Now, the Hamilton-Cartan form $\Omega_h$ has the local expression:
$$
\Omega_{h_{\mathfrak V}}=\d H_{\mathfrak V}\wedge \d^4x-\sum_{\alpha\leq\beta}\d p_g^{\alpha\beta,\mu}\wedge\d g_{\alpha\beta}\wedge\d^3x_\mu-\d p_{\phi}^{,\mu}\wedge \d \phi\wedge \d^3x_{\mu}-\sum_{\alpha\leq\beta}\d L_g^{\alpha\beta,\mu\nu}\wedge\d V_{\alpha\beta,\mu}\wedge\d^3x_\nu - \d L_{\phi}^{,\mu\nu}\wedge \d U_{,\mu}\wedge \d^3x_{\nu}\ ,
$$
and the local expression of a representative of a class $\{{\bf X}_h\}$ of 
semi-holonomic multivector fields in ${\cal P}$ is
$$
{\bf X}_h=\bigwedge_{i=\nu}^4\left(\frac{\partial}{\partial x^\nu}+F_{g\;\alpha\beta,\nu}\frac{\partial}{\partial g_{\alpha\beta}}+F_{\phi\;,\nu}\frac{\partial}{\partial \phi}+G^{\alpha\beta,\mu}_{g\;\nu}\frac{\partial}{\partial p_g^{\alpha\beta,\mu}}+G^{,\mu}_{\phi\;\nu}\frac{\partial}{\partial p_{\phi}^{,\mu}}\right);.
$$
with $F_{g\;\alpha\beta,\nu}$, $G^{\alpha\beta,\mu}_{g\;\nu}$, $F_{\phi\;,\nu}$,  $G^{,\mu}_{\phi\;\nu}\in C^\infty(\mathcal{P})$.

From \eqref{eq:hvf} we obtain
\begin{align}\label{eq:hamiltonparticularcase1}
\frac{\partial H_{\mathfrak V}}{\partial g_{\alpha\beta}}=&-G ^{\alpha\beta,\mu}_{g\;\mu}+G^{\rho\lambda,\gamma}_{g\;\nu}\frac{\partial V_{\delta\zeta,\epsilon}}{\partial p_{g}^{\rho\lambda,\gamma}}\frac{\partial L^{\delta\zeta,\epsilon\nu}}{\partial g_{\alpha\beta}}+F_{g\;\rho\lambda,\nu}\left(\frac{\partial V_{\delta\zeta,\epsilon}}{\partial g_{\rho\lambda}}\frac{\partial L^{\delta\zeta,\epsilon\nu}}{\partial g_{\alpha\beta}}-\frac{\partial V_{\delta\zeta,\epsilon}}{\partial g_{\alpha\beta}}\frac{\partial L^{\delta\zeta,\epsilon\nu}}{\partial g_{\rho\lambda}}\right)
\nonumber\\
&+G_{\phi\;\rho}^{,\rho}\left(\frac{\partial U_{,\nu}}{\partial p_{\phi}^{,\mu}}\frac{\partial L_{\phi}^{,\mu\nu}}{\partial g_{\alpha\beta}}-\frac{\partial U_{,\nu}}{\partial g_{\alpha\beta}}\frac{\partial L_{\phi}^{,\mu\nu}}{\partial p_{\phi}^{,\mu}}\right)+F_{\phi\;,\mu}\left(\frac{\partial U_{,\nu}}{\partial \phi}\frac{\partial L_{\phi}^{,\mu\nu}}{\partial g_{\alpha\beta}}-\frac{\partial U_{,\nu}}{\partial g_{\alpha\beta}}\frac{\partial L_{\phi}^{,\mu\nu}}{\partial \phi}\right),
\\
\frac{\partial H_{\mathfrak V}}{\partial p^{\alpha\beta,\mu}}=&F_{g\;\alpha\beta,\mu}-F_{g\;\rho\lambda,\nu}\frac{\partial V_{ab,c}}{\partial p_g^{\alpha\beta,\mu}}\frac{\partial L_g^{ab,c\nu}}{\partial g_{\rho\lambda}},\\
\frac{\partial H_{\mathfrak V}}{\partial \phi}=&-G_{\phi\;\nu}^{,\rho}\left(\delta_{\rho}^{\nu}+\frac{\partial U_{,\mu}}{\partial \phi}\frac{\partial L_{\phi}^{,\mu\nu}}{\partial p_{\phi}^{,\rho}}-\frac{\partial U_{,\mu}}{\partial p_{\phi}^{,\rho}}\frac{\partial L_{\phi}^{,\mu\nu}}{\partial \phi}\right)\nonumber\\
&+F_{g\;\alpha\beta,\nu}\left(\frac{\partial V_{\rho\lambda,\mu}}{\partial \phi}\frac{\partial L_{g}^{\alpha\beta,\mu\nu}}{\partial g_{\rho\lambda}}+\frac{\partial U_{,\mu}}{\partial g_{\alpha\beta}}\frac{\partial L_{\phi}^{,\mu\nu}}{\partial \phi}-\frac{\partial U_{,\mu}}{\partial \phi}\frac{\partial L_{\phi}^{,\mu\nu}}{\partial g_{\alpha\beta}}\right),\\\label{eq:hamiltonparticularcase2}
\frac{\partial H_{\mathfrak V}}{\partial p_{\phi}^{,\mu}}=&F_{\phi\;,\mu}+F_{g \; \alpha\beta,\nu}\left(\frac{\partial U_{,\rho}}{\partial g_{\alpha\beta}}\frac{\partial L_{\phi}^{,\rho\nu}}{\partial p_{\phi}^{,\mu}}-\frac{\partial U_{,\rho}}{\partial p_{\phi}^{,\mu}}\frac{\partial L_{\phi}^{,\rho\nu}}{\partial g_{\alpha\beta}}-\frac{\partial V_{\delta\zeta,\epsilon}}{\partial p_{\phi}^{,\mu}}\frac{\partial L_{g}^{\delta\zeta,\epsilon\nu}}{\partial g_{\alpha\beta}}\right)\nonumber\\
&+F_{\phi\;,\nu}\left(\frac{\partial U_{,\rho}}{\partial \phi}\frac{\partial L_{\phi}^{,\rho\nu}}{\partial p_{\phi}^{,\mu}}-\frac{\partial U_{,\rho}}{\partial p_{\phi}^{\mu}}\frac{\partial L_{\phi}^{,\rho\nu}}{\partial \phi}\right).
\end{align}

Expresions (\ref{eq:hamiltonparticularcase1}) through (\ref{eq:hamiltonparticularcase2}) would be the classical Hamilton-De Donder-Weil equations
for a first order field theory except by the fact that they contain 
extra-terms because the cubic Horndeski Lagrangian is a  second order theory with respect to the metric and the scalar field, and neither
$\displaystyle L_g^{\alpha\beta,\mu\nu}=\frac{1}{n(\mu\nu)}\frac{\partial L}{\partial g_{\alpha\beta,\mu\nu}}$ nor $L_{\phi}^{,\mu\nu}=\frac{1}{n(\mu\nu)}\frac{\partial L}{\partial \phi_{;\mu\nu}}$
vanish.

\subsection{Hamiltonian formalism for the general case}

On this subsection we consider the full cubic Horndeski Lagrangian 

\begin{equation}
L_{\mathfrak V}=\frac{1}{16\pi G}\sqrt{|g|}\left[R+X+G_2(\phi,X)+G_3(\phi,X)\square\phi\right].
\end{equation}

We will assume that $L$ is almost-regular. The multimomenta for the general cubic case are

\begin{align}\label{siligadura}
    p_g^{\alpha\beta,\mu}=&-\frac{1}{2}\sqrt{-g}\Big\{G_3(\phi,X)\phi_{;\gamma}\left(g^{\gamma\beta}g^{\alpha\mu}+g^{\gamma\alpha}g^{\beta\mu}-g^{\gamma\mu}g^{\alpha\beta}\right)\nonumber\\
    &-g_{\rho\lambda,\sigma}\Big[-3g^{\rho\sigma}g^{\mu(\alpha}g^{\beta) \lambda}+2g^{\rho\lambda}g^{\mu(\alpha}g^{\beta)\sigma}+2g^{\alpha\beta}g^{\rho\sigma}g^{\mu \lambda}\nonumber\\
    &+3g^{\mu \sigma}g^{\lambda(\alpha}g^{\beta)\rho}-2g^{\rho\mu}g^{\lambda(\alpha}g^{\beta)\sigma}-g^{\alpha\beta}g^{\rho\lambda}g^{\mu \sigma}\nonumber\\
    &-\frac{1}{2}g^{\rho\lambda}g^{\mu \sigma}g^{\alpha\beta}+g^{\mu(\rho}g^{\lambda)\sigma}g^{\alpha\beta}+g^{\mu \sigma}g^{\alpha(\rho}g^{\lambda)\beta}\Big]\Big\}\,,
\end{align}

\begin{align}\label{noligadura}
    p_{\phi}^{,\mu}=&-\sqrt{-g}g^{\mu\nu}\Bigg\{\phi_{;\nu}\left[1+\frac{\partial G_2(\phi,X)}{\partial X}+\frac{\partial G_3(\phi,X)}{\partial \phi}+\square\phi\frac{\partial G_3(\phi,X)}{\partial X}\right]\nonumber\\
    &+\left(\phi_{;\alpha\nu}+\phi_{;\gamma}\Gamma^{\gamma}_{\alpha\nu}\right)g^{\alpha\beta}\phi_{;\beta}\frac{\partial G_3(\phi,X)}{\partial X}\Bigg\}=L_{\phi_h}^\mu.
\end{align}

The multimomentum (\ref{noligadura}) is not a constraint for the general case, in which $G_3\neq 0$. In contrast, (\ref{siligadura}) is indeed a constraint and we must demand tangency of the multivector field to the submanifold defined by this constraint. 

\begin{equation*}
   \Lie(X_{h\;\tau})\left(p_{\phi}^{,\mu}-L_{\phi_h}^\mu\right)\vert_{\W_c}=0,
\end{equation*}

which yields

\begin{align}\label{tangenciaconstraint}
0=&G_{g\;\tau}^{\alpha\beta,\mu}-\frac{1}{2}\sqrt{-g}\Bigg\{g^{\rho\lambda}F_{g\;\;\rho\lambda,\tau}N^{\alpha\beta\mu}-G_3(\phi,X)F_{g\;\rho\lambda,\tau}\phi_{;\gamma}\frac{\partial}{\partial g_{\rho\lambda}}\left(g^{\gamma\beta}g^{\alpha\mu}+g^{\gamma\alpha}g^{\beta\mu}-g^{\gamma\mu}g^{\alpha\beta}\right)\nonumber\\
&-F_{\phi,\tau}\frac{\partial G_3}{\partial \phi}\phi_{;\gamma}\left(g^{\gamma\beta}g^{\alpha\mu}+g^{\gamma\alpha}g^{\beta\mu}-g^{\gamma\mu}g^{\alpha\beta}\right)-g_{\rho\lambda,\sigma}F_{g\;\delta\zeta,\tau}\frac{\partial M^{\rho\lambda \sigma\alpha\beta\mu}}{\partial g_{\delta\zeta}}+F_{g\;\delta\zeta,\eta\tau}\delta_{\rho\lambda}^{\delta\zeta}\delta_{\sigma}^{\eta}M^{\rho\lambda \sigma\alpha\beta\mu}\nonumber\\
&-F_{\phi\;\rho,\tau}\delta_{\gamma}^{\;\rho}\left(g^{\gamma\beta}g^{\alpha\mu}+g^{\gamma\alpha}g^{\beta\mu}-g^{\gamma\mu}g^{\alpha\beta}\right)
\Bigg\},
\end{align}

where 

\begin{align*}
M^{\rho\lambda\sigma\alpha\beta\mu}=&-3g^{\rho\sigma}g^{\mu(\alpha}g^{\beta) \lambda}+2g^{\rho\lambda}g^{\mu(\alpha}g^{\beta)\sigma}+2g^{\alpha\beta}g^{\rho\sigma}g^{\mu\lambda}\\
     &+3g^{\mu \sigma}g^{\lambda(\alpha}g^{\beta)\rho}-2g^{\rho\mu}g^{\lambda(\alpha}g^{\beta)\sigma}-g^{\alpha\beta}g^{\rho\lambda}g^{\mu \sigma}\\
    &-\frac{1}{2}g^{\rho\lambda}g^{\mu \sigma}g^{\alpha\beta}+g^{\mu(\rho}g^{\lambda)\sigma}g^{\alpha\beta}+g^{\mu \sigma}g^{\alpha(\rho}g^{\lambda)\beta}
\end{align*}

and

\begin{equation*}
N^{\alpha\beta\mu}=G_3(\phi,X)\phi_{;\gamma}\left(g^{\gamma\beta}g^{\alpha\mu}+g^{\gamma\alpha}g^{\beta\mu}-g^{\gamma\mu}g^{\alpha\beta}\right)-g_{\rho\lambda,\sigma}M^{\rho\lambda\sigma\alpha\beta\mu}.
\end{equation*}

The second-order multimomenta $p_g^{\alpha\beta,\mu\nu}$ and $p_{\phi}^{,\mu\nu}$ are completely determined by the constrains defining $\mathcal{P}$ \ref{prop:GraphLegMapSect}. It is not possible, in general, to isolate the velocities in terms of the momenta unless $G_2(\phi,X)$ and $G_3(\phi,X)$ are explicitly specified. Moreover, the momenta depend on the acceleration of the scalar field, i.e. $\mathcal{P}$ does not project on $J^1\pi$ as expected \ref{prop:proyecta}. It is impossible, in general, to explicitly isolate the velocities purely in terms of the momenta, so we will use the mix coordinates $   (x^\mu,g_{\alpha\beta},\phi,g_{\alpha\beta,\mu},\phi_{;\mu},p_g^{\alpha\beta,\mu},p_{\phi}^{,\mu},p_g^{\alpha\beta,\mu\nu},p_{\phi}^{,\mu\nu})$ for the Hamiltonian formulation, which contains a positions, momenta and velocities. 

In terms of these coordinates, the Hamiltonian function is

\begin{align}
     H_{\mathfrak V}=&\sum_{\alpha\leq\beta}p_g^{\alpha\beta,\mu}g_{\alpha\beta,\mu}+\sqrt{-g}G_3(\phi,X)g^{\mu\nu}\phi_{;\gamma}\Gamma_{\nu\mu}^{\gamma}+p_{\phi}^{,\mu}\phi_{;\mu}-\sqrt{-g}\left(X+G_2(\phi,X)\right)\nonumber\\
    &-\sqrt{-g}g^{ab}\left[-\frac{1}{2}g^{ce}g^{df}g_{ef,c}\left(g_{bd,a}+g_{ad,b}-g_{ab,d}\right)+\frac{1}{2}g^{ce}g^{df}g_{ef,b}\left(g_{ad,c}+g_{cd,a}-g_{ac,d}\right)+\Gamma_{cd}^c\Gamma_{ba}^{d}-\Gamma_{ba}^{c}\Gamma_{ca}^{d}\right]
\end{align}

The Hamilton-Cartan form $\Omega_h$ has the local expression:
$$
\Omega_{h_{\mathfrak V}}=\d H_{\mathfrak V}\wedge \d^4x-\sum_{\alpha\leq\beta}\d p_g^{\alpha\beta,\mu}\wedge\d g_{\alpha\beta}\wedge\d^3x_\mu-\d p_{\phi}^{,\mu}\wedge \d \phi\wedge \d^3x_{\mu}-\sum_{\alpha\leq\beta}\d L_g^{\alpha\beta,\mu\nu}\wedge\d g_{\alpha\beta,\mu}\wedge\d^3x_\nu - \d L_{\phi}^{,\mu\nu}\wedge \d\phi_{;\mu}\wedge \d^3x_{\nu}\ ,
$$
and the local expression of a representative of a class $\{{\bf X}_h\}$ of 
semi-holonomic multivector fields in ${\cal P}$ is
$$
{\bf X}_h=\bigwedge_{i=\nu}^4\left(\frac{\partial}{\partial x^\nu}+F_{g\;\alpha\beta,\nu}\frac{\partial}{\partial g_{\alpha\beta}}+F_{\phi\;,\nu}\frac{\partial}{\partial \phi}+F_{g\;\alpha\beta,\mu,\nu}\frac{\partial}{\partial g_{\alpha\beta,\mu}}+F_{\phi\;\mu,\nu}\frac{\partial}{\partial \phi_{;\mu}}+G^{\alpha\beta,\mu}_{g\;\nu}\frac{\partial}{\partial p_g^{\alpha\beta,\mu}}+G^{,\mu}_{\phi\;\nu}\frac{\partial}{\partial p_{\phi}^{,\mu}}\right),
$$
with $F_{g\;\alpha\beta,\nu}$, $G^{\alpha\beta,\mu}_{g\;\nu}$,$F_{g\;\alpha\beta,\mu,\nu}$, $F_{\phi\;,\nu}$, $F_{\phi\;,\mu,\nu}$,$G^{,\mu}_{\phi\;\nu}$ $\in C^\infty(\mathcal{P})$.

From \eqref{eq:hvf} we get 

\begin{align}
\frac{\partial H_{\mathfrak V}}{\partial g_{\alpha\beta}}=&-G ^{\alpha\beta,\mu}_{g\;\mu}+F_{g\;ab,\mu,\nu}\frac{\partial L_{g}^{ab,\mu\nu}}{\partial g_{\alpha\beta}}+F_{\phi\;\mu,\nu}\frac{\partial L_{\phi}^{,\mu\nu}}{\partial g_{\alpha\beta}}\label{eqn:hamiltoninicio}
\\
\frac{\partial H_{\mathfrak V}}{\partial g_{\alpha\beta,\mu}}=&F_{g\;ab,\nu}\frac{\partial L_g^{\alpha\beta,\mu\nu}}{\partial g_{ab}}\\
\frac{\partial H_{\mathfrak V}}{\partial p_g^{\alpha\beta,\mu}}=&F_{g\;\alpha\beta,\nu}\\
\frac{\partial H_{\mathfrak V}}{\partial \phi}=&-G_{\phi\;\mu}^{,\mu}+F_{\phi\;\mu\nu}\frac{\partial L_{\phi}^{,\mu\nu}}{\partial \phi}\\
\frac{\partial H_{\mathfrak V}}{\partial \phi_{;\mu}}=&F_{\phi\;,\nu}\frac{\partial L_{\phi}^{,\mu\nu}}{\partial \phi}-g^{\rho\sigma}\phi_{;\sigma}F_{\phi\;\rho,\nu}\frac{\partial L_{\phi}^{,\mu\nu}}{\partial X}+F_{g\;ab,\nu}\frac{\partial L_{\phi}^{,\mu\nu}}{\partial g_{ab}}\\
\frac{\partial H_{\mathfrak V}}{\partial p_{\phi}^{,\mu}}=&F_{\phi\;,\mu}\label{eqn:hamiltonfin}.
\end{align}

These are the covariant Hamilton equations for the cubic Horndeski's theory. Depending on the values of $G_2(\phi,X)$ and $G_3(\phi,X)$, these equations may not be compatible and the constraint algorithm should be continued. The dynamics of the Hamiltonian theory are determined by the Hamilton equations \ref{eqn:hamiltoninicio} to (\ref{eqn:hamiltonfin}) and the tangency condition of the multivector field to the constraint (\ref{tangenciaconstraint}).


\section{Conclusion}
\label{Section6}

In this work, we presented a multisymplectic covariant description of the cubic Horndeski theory using the unified Lagrangian-Hamiltonian formalism. The constraint algorithm was employed to determine a submanifold of the higher-order jet-multimomentum bundle $\mathcal{W}_r$ and the corresponding constraints that provide the main features of the theory. 

The constraints (\ref{eqn:UniVec5}), (\ref{eqn:UniVec5}), (\ref{Wlag1}) and (\ref{Wlag2}) appear as a consequence of this formalism and define the Legendre map which further allows to pose a covariant Hamiltonian formulation and the corresponding Hamilton-de Donder-Weyl-like equations of the theory. Although more constraints appear, they have no physical relevance and are a mere consequence of the projectability of the theory, as we would expect a second order Lagrangian to produce fourth-order equations of motion, but it does produce second-order equations of motion. 

We showed that the Poincaré-Cartan form of the theory form does not necessarily project onto $J^1\pi$, unless $\frac{\partial G_3(\phi,X)}{\partial X}=0$. This makes it impossible, in general, to obtain a covariant Hamiltonian formulation with first order equations of motion. Moreover, this is a counterexample that proves that the projectability of the equations does not implies the projectability of the geometric structures. Hence, the reciprocal of proposition 1 in \cite{gaset_order_2016} does not hold.

With extra assumptions on the Lagrangian, we provide the expression of the velocities in function of the momenta, providing a covariant formulation of the Hamiltonian formalism which involves only multimomenta. We also present situation where this is not possible. 

For a general cubic Horndeski's theory, we provide the covariant Hamiltonian formalism and we present the field equations. In general, they involve velocities of the metric and the scalar field, as well as the accelerations of the scalar field.


It is yet to be seen the map between our covariant formulation and the instantaneous, or ADM-like, Hamiltonian formulation of these theories already presented in \cite{Kovacs1}. Recently, a proof of the equivalence of the symplectic forms derived from the canonical and the covariant phase space formalisms was presented in \cite{Margalef}. It is yet to be proven if there's a similar equivalence between the symplectic forms derived from the instantaneous and the multisymplectic forms in the multisymplectic formalism. If they are equivalent, it will be needed to determine how to map all these forms. This  will be explored in a future work.



\section*{Acknowledgments}

The authors acknowledge financial support from the Ministerio de Ciencia, Innovaci\'on y Universidades (Spain), projects PGC2018-098265-B-C33 and D2021-125515NB-21.


\end{document}